\documentclass[12pt,draftclsnofoot, onecolumn]{IEEEtran}

\normalsize

\usepackage[T1]{fontenc}
\usepackage{color}
\usepackage[protrusion=true,expansion=true]{microtype}	
\usepackage{amsmath,amsfonts,amsthm} 
\usepackage[pdftex]{graphicx}	
\usepackage{url}
\usepackage{float}

\usepackage{textcomp}
\usepackage{algorithmic}
\usepackage{algorithm}

\usepackage{fancyhdr}
\newtheorem{definition}{Definition}
\newtheorem{remark}{Remark}
\newtheorem{theorem}{Theorem}
\newtheorem{claim}{Claim}
\newtheorem{lemma}{Lemma}
\usepackage{amssymb}

\usepackage{physics}
\usepackage{amsmath}
\usepackage{tikz}
\usepackage{mathdots}
\usepackage{yhmath}
\usepackage{cancel}
\usepackage{color}
\usepackage{siunitx}
\usepackage{array}
\usepackage{multirow}
\usepackage{amssymb}
\usepackage{gensymb}
\usepackage{tabularx}
\usepackage{booktabs}
\usepackage{mathtools}

\usetikzlibrary{fadings}
\usetikzlibrary{patterns}
\usetikzlibrary{shadows.blur}
\usetikzlibrary{shapes}
\usepackage[numbers,sort&compress]{natbib}

\definecolor{auburn}{rgb}{0.43, 0.21, 0.1}

\title{On the Stability Region of Intermittent Interference Networks}

\author{Sajjad~Nassirpour
        and~Alireza~Vahid
        \thanks{The authors are with the Department of Electrical Engineering at the University of Colorado, Denver, USA. Email: {\sffamily sajjad.nassirpour@ucdenver.edu} and {\sffamily alireza.vahid@ucdenver.edu}.}
\thanks{The preliminary results of this work were presented at the 2020 IEEE Computing and Communication Workshop and Conference (CCWC)~\cite{nassirpour2020throughput}.
}
}

\renewcommand\IEEEkeywordsname{Index Terms}

\begin{document}
\maketitle
\begin{abstract}
Recent information-theoretic studies have resulted in several interference management (IM) techniques that promise significant capacity improvements over interference avoidance techniques.
However, in practice, the stable throughput region is a more relevant metric compared to the capacity region.
In this work, we focus on the stable throughput region of a two-pair intermittent interference 
 network with distributed transmitters and propose a queue-based transmission protocol in different regimes to handle the data between queues.
In this context, we translate physical-layer IM protocols to
accommodate stochastic message arrivals. 
To evaluate our proposed techniques, we compare the stable throughput region to the capacity region and show, through simulations, that the stable throughput region matches the capacity region, when the latter is known.
We show that in order to achieve the optimal stable throughput region, new ingredients are needed when compared to prior results. We quantify the trade-off between encoding/decoding complexity of the proposed scheme (in terms of number of required algebraic operations), and the  achievable rates. Finally, we study the lifetime of messages (i.e. the duration from arrival to successful delivery) vis-a-vis the total communication time, and we observe that the average lifetime scales as the square root of the total communication time.
\end{abstract}


\section{Introduction}
\label{section:intro}
As wireless networks grow in size, the number of simultaneous transmissions increases resulting in higher chance of signal interference. In fact, interference has become the main bottleneck for network throughput improvements in wireless systems. Therefore, Interference Management (IM) techniques play a central role in enhancing throughput rates, and in recent years, we have seen a variety of IM techniques in the Information Theory literature. Of particular interest are networks with intermittent connectivity as they model last-mile communications, shadowing in mmWave, fog-aided systems, bursty and packet transmission.


In this work, we focus on the scenario in which messages arrive stochastically at the transmitters, and thus, introduce new challenges in interference management. Stable throughput region of multi packet reception channel model under ALOHA mechanism is described in~\cite{mergen2005stability,luo2006throughput}. Authors in~\cite{cadambe2008duality} characterize the stable throughput region of the broadcast channel (BC) and the multiple-access channel (MAC) with backlogged traffic. The stable throughput region of the two-user BC with feedback is characterized in~\cite{sagduyu2012capacity,listability}. Authors in~\cite{wang2012capacity} show that the complexity of analyzing the capacity and the stability regions of $K$-user ($K \geq 4$) BCs grows rapidly with the number of users, and tracking all interfering signals becomes a daunting task. These results, for the most part, consider a central transmitter that has access to all the messages in the network. 


We consider a distributed intermittent interference network with two transmitter-receiver pairs. In this setting, transmitters do not exchange any messages and thus, we ought to make decentralized and distributed decisions for interference management. For this problem, we study the stable throughput region. We show, through simulations, that the stable throughput region matches the capacity region when the latter is known. 
Motivated by recent IM techniques for networks with delayed channel state information at the transmitters (CSIT) which is a more suitable model for dynamic networks (\emph{e.g.}~\cite{maddah2012completely,viswanathan1999capacity,vaze2011degrees}), we consider wireless networks with low feedback bandwidth where only a couple of feedback bits are available per transmission. We note that each data packet in the forward channel may contain thousands of bits and thus, feedback overhead is negligible. In intermittent networks, transmitters send out data packets and receive delayed ACK/NACK messages (\emph{i.e.} indicating successful delivery or failure).


There are several communication protocols that achieve the information-theoretic bounds on throughput rates of networks with ACK/NACK feedback. However, these results are again mostly limited to centralized transmitters. XOR combination technique, which incorporates feedback from receivers is introduced in~\cite{jolfaei1993new}, to reduce the retransmission attempts and improve throughput rate under repetitive Automatic Repeat reQuest (ARQ) strategy in multi-users BCs. Authors in~\cite{georgiadis2009broadcast} show that by employing the XOR operation and random coding mechanism, the upper bound of the capacity region in BC model is achievable. A threshold-based scheduling procedure is defined in~\cite{reddy2012distributed} which describes the policies of threshold value determination to handle the packets and maximize the throughput rate in $K$-user interference wireless network. 

\begin{figure}[!ht]
  \centering

\tikzset{every picture/.style={line width=0.75pt}} 

\begin{tikzpicture}[x=0.75pt,y=0.75pt,yscale=-1,xscale=1]

\draw    (119,124) -- (230.17,124.33) ;
\draw    (119,143) -- (230.17,143.33) ;
\draw    (230.17,124.33) -- (230.17,143.33) ;
\draw    (210.17,125.33) -- (210.17,142.33) ;
\draw    (190.17,125.33) -- (190.17,142.33) ;
\draw    (170.17,125.33) -- (170.17,142.33) ;
\draw    (150.17,125.33) -- (150.17,142.33) ;
\draw    (130.17,125.33) -- (130.17,142.33) ;
\draw    (119,214) -- (230.17,214.33) ;
\draw    (119,233) -- (230.17,233.33) ;
\draw    (230.17,214.33) -- (230.17,233.33) ;
\draw    (210.17,215.33) -- (210.17,232.33) ;
\draw    (190.17,215.33) -- (190.17,232.33) ;
\draw    (170.17,215.33) -- (170.17,232.33) ;
\draw    (150.17,215.33) -- (150.17,232.33) ;
\draw    (130.17,215.33) -- (130.17,232.33) ;
\draw  [fill={rgb, 255:red, 155; green, 155; blue, 155 }  ,fill opacity=1 ] (190.17,124.1) -- (210.17,124.1) -- (210.17,143.1) -- (190.17,143.1) -- cycle ;
\draw  [fill={rgb, 255:red, 155; green, 155; blue, 155 }  ,fill opacity=1 ] (150.17,124.33) -- (170.17,124.33) -- (170.17,143.33) -- (150.17,143.33) -- cycle ;
\draw  [fill={rgb, 255:red, 155; green, 155; blue, 155 }  ,fill opacity=1 ] (210.17,124.33) -- (230.17,124.33) -- (230.17,143.33) -- (210.17,143.33) -- cycle ;
\draw  [fill={rgb, 255:red, 155; green, 155; blue, 155 }  ,fill opacity=1 ] (130.17,214.33) -- (150.17,214.33) -- (150.17,233.33) -- (130.17,233.33) -- cycle ;
\draw  [fill={rgb, 255:red, 155; green, 155; blue, 155 }  ,fill opacity=1 ] (190.17,214.33) -- (210.17,214.33) -- (210.17,233.33) -- (190.17,233.33) -- cycle ;
\draw    (240.57,134) -- (257.55,134.21) -- (278.57,134.46) ;
\draw [shift={(281.57,134.5)}, rotate = 180.7] [fill={rgb, 255:red, 0; green, 0; blue, 0 }  ][line width=0.08]  [draw opacity=0] (8.93,-4.29) -- (0,0) -- (8.93,4.29) -- cycle    ;
\draw    (238.57,224) -- (255.55,224.21) -- (276.57,224.46) ;
\draw [shift={(279.57,224.5)}, rotate = 180.7] [fill={rgb, 255:red, 0; green, 0; blue, 0 }  ][line width=0.08]  [draw opacity=0] (8.93,-4.29) -- (0,0) -- (8.93,4.29) -- cycle    ;
\draw   (287.5,123.71) -- (337.57,123.71) -- (337.57,153.71) -- (287.5,153.71) -- cycle ;
\draw    (338.57,140.5) -- (347.57,140.5) ;
\draw    (347.57,140.5) -- (347.57,120.5) ;
\draw    (347.57,120.5) -- (357.57,110.5) ;
\draw    (338.57,110.5) -- (357.57,110.5) ;
\draw    (347.57,120.5) -- (338.57,110.5) ;
\draw   (288.5,213.71) -- (338.57,213.71) -- (338.57,244.71) -- (288.5,244.71) -- cycle ;
\draw    (339.57,229.5) -- (348.57,229.5) ;
\draw    (348.57,229.5) -- (348.57,209.5) ;
\draw    (348.57,209.5) -- (358.57,199.5) ;
\draw    (339.57,199.5) -- (358.57,199.5) ;
\draw    (348.57,209.5) -- (339.57,199.5) ;
\draw   (498.5,123.71) -- (548.57,123.71) -- (548.57,154.71) -- (498.5,154.71) -- cycle ;
\draw    (488.57,140.5) -- (497.57,140.5) ;
\draw    (488.57,140.5) -- (488.57,120.5) ;
\draw    (488.57,120.5) -- (498.57,110.5) ;
\draw    (479.57,110.5) -- (498.57,110.5) ;
\draw    (488.57,120.5) -- (479.57,110.5) ;
\draw   (497.5,214.71) -- (547.57,214.71) -- (547.57,244.71) -- (497.5,244.71) -- cycle ;
\draw    (487.57,230.5) -- (496.57,230.5) ;
\draw    (487.57,230.5) -- (487.57,210.5) ;
\draw    (487.57,210.5) -- (497.57,200.5) ;
\draw    (478.57,200.5) -- (497.57,200.5) ;
\draw    (487.57,210.5) -- (478.57,200.5) ;
\draw    (354.57,118.5) -- (364.5,118.71) -- (371.55,118.71) -- (437.5,118.71) ;
\draw [shift={(440.5,118.71)}, rotate = 180.01] [fill={rgb, 255:red, 0; green, 0; blue, 0 }  ][line width=0.08]  [draw opacity=0] (8.93,-4.29) -- (0,0) -- (8.93,4.29) -- cycle    ;
\draw    (355.57,209.5) -- (365.5,209.71) -- (372.55,209.71) -- (438.5,209.71) ;
\draw [shift={(441.5,209.71)}, rotate = 180.01] [fill={rgb, 255:red, 0; green, 0; blue, 0 }  ][line width=0.08]  [draw opacity=0] (8.93,-4.29) -- (0,0) -- (8.93,4.29) -- cycle    ;
\draw    (357.5,205.71) -- (430.19,136.49) -- (439.33,127.78) ;
\draw [shift={(441.5,125.71)}, rotate = 496.4] [fill={rgb, 255:red, 0; green, 0; blue, 0 }  ][line width=0.08]  [draw opacity=0] (8.93,-4.29) -- (0,0) -- (8.93,4.29) -- cycle    ;
\draw    (354.57,123.5) -- (437.31,200.67) ;
\draw [shift={(439.5,202.71)}, rotate = 223.01] [fill={rgb, 255:red, 0; green, 0; blue, 0 }  ][line width=0.08]  [draw opacity=0] (8.93,-4.29) -- (0,0) -- (8.93,4.29) -- cycle    ;
\draw    (463,120) -- (464.51,120.09) -- (479.5,119.78) ;
\draw [shift={(482.5,119.71)}, rotate = 538.8199999999999] [fill={rgb, 255:red, 0; green, 0; blue, 0 }  ][line width=0.08]  [draw opacity=0] (8.93,-4.29) -- (0,0) -- (8.93,4.29) -- cycle    ;
\draw    (463,211) -- (479.5,210.75) ;
\draw [shift={(482.5,210.71)}, rotate = 539.29] [fill={rgb, 255:red, 0; green, 0; blue, 0 }  ][line width=0.08]  [draw opacity=0] (8.93,-4.29) -- (0,0) -- (8.93,4.29) -- cycle    ;

\draw (92,125) node [anchor=north west][inner sep=0.75pt]    {$\lambda _{1}$};
\draw (92,215) node [anchor=north west][inner sep=0.75pt]    {$\lambda _{2}$};
\draw (300,131) node [anchor=north west][inner sep=0.75pt]    {${\sf Tx}_1$};
\draw (300,221) node [anchor=north west][inner sep=0.75pt]    {${\sf Tx}_2$};
\draw (510,131) node [anchor=north west][inner sep=0.75pt]    {${\sf Rx}_1$};
\draw (510,221) node [anchor=north west][inner sep=0.75pt]    {${\sf Rx}_2$};
\draw (440,110) node [anchor=north west][inner sep=0.75pt]  [font=\LARGE]  {$\oplus $};
\draw (440,200) node [anchor=north west][inner sep=0.75pt]  [font=\LARGE]  {$\oplus $};
\draw (375,100) node [anchor=north west][inner sep=0.75pt]  [font=\small]  {$S_{11}[t]$};
\draw (375,215) node [anchor=north west][inner sep=0.75pt]  [font=\small]  {$S_{22}[t]$};
\draw (370,117) node [anchor=north west][inner sep=0.75pt]  [font=\small,rotate=-41.74]  {$S_{12}[t]$};
\draw (360,181.53) node [anchor=north west][inner sep=0.75pt]  [font=\small,rotate=-317.57]  {$S_{21}[t]$};
\draw (475,90) node [anchor=north west][inner sep=0.75pt]    {$Y_{1}[t]$};
\draw (475,180) node [anchor=north west][inner sep=0.75pt]    {$Y_{2}[t]$};

\draw (325,90) node [anchor=north west][inner sep=0.75pt]    {$X_{1}[t]$};
\draw (325,180) node [anchor=north west][inner sep=0.75pt]    {$X_{2}[t]$};

\draw (255.5,88) node [anchor=north west][inner sep=0.75pt]  [color={rgb, 255:red, 208; green, 2; blue, 27 }  ,opacity=1 ]  {$S^{t-1}$};
\draw (255.5,257) node [anchor=north west][inner sep=0.75pt]  [color={rgb, 255:red, 208; green, 2; blue, 27 }  ,opacity=1 ]  {$S^{t-1}$};

\draw [color={rgb, 255:red, 208; green, 2; blue, 27 }  ,draw opacity=1 ]   (320.5,100.71) -- (320.5,117) ;
\draw [shift={(320.5,120)}, rotate = 270] [fill={rgb, 255:red, 208; green, 2; blue, 27 }  ,fill opacity=1 ][line width=0.08]  [draw opacity=0] (8.93,-4.29) -- (0,0) -- (8.93,4.29) -- cycle    ;
\draw [color={rgb, 255:red, 208; green, 2; blue, 27 }  ,draw opacity=1 ]   (290.5,100.71) -- (320.5,100.71) ;
\draw [color={rgb, 255:red, 208; green, 2; blue, 27 }  ,draw opacity=1 ]   (320.5,268.71) -- (320.5,253.71) ;
\draw [shift={(320.5,250.71)}, rotate = 450] [fill={rgb, 255:red, 208; green, 2; blue, 27 }  ,fill opacity=1 ][line width=0.08]  [draw opacity=0] (8.93,-4.29) -- (0,0) -- (8.93,4.29) -- cycle    ;
\draw [color={rgb, 255:red, 208; green, 2; blue, 27 }  ,draw opacity=1 ]   (290.5,268.71) -- (320.5,268.71) ;
\end{tikzpicture}
  \caption{A two-user binary intermittent interference network with stochastic data arrivals.}
  \label{Fig:network_model}
\end{figure}
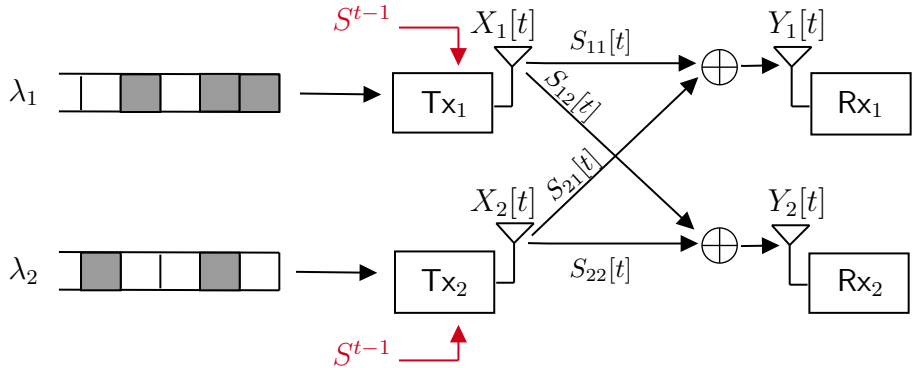


\noindent\textbf{Related Work}: The intermittent interference channel abstraction that we use in this work, was first introduced in~\cite{vahid2011interference}, and the capacity region of the two-user setting is investigated under various assumptions~\cite{vahid2014communication,vahid2014capacity,vahid2012binary,vahid2015impact,vahid2016does,vahid2016two,vahid2017binary,vahid2018arq}. 
This model consists of two transmitter-receiver pairs, channel connectivity is governed by binary coefficients $S_{ij}[t]$'s at time $t$, and we assume these coefficients are available to the transmitters with unit delay. We study the stable throughput region of this problem as opposed to the capacity region analysis of~\cite{vahid2014capacity,vahid2014communication}. 
 
In~\cite{pan2016stability}, the authors study a similar problem, deploy similar coding ideas to~\cite{vahid2014capacity}, and propose a modified version to accommodate stochastic packet arrivals. Unfortunately, there are a number of errors in coding strategy in that paper which we will discuss in Remark~\ref{remark:fault_example} of Section~\ref{Subsection:DO}. In~\cite{papas2019adhoc}, the authors derive stable throughput region but ignore the potential benefits of coding opportunities across transmitters, and thus, the network is treated as interfering broadcast channels for which the achievable rates are inferior to the interference channel.

\noindent\textbf{Contributions}: Our main contributions are as follows:
\begin{enumerate}

\item We extend the outer-bounds in~\cite{vahid2014capacity} to non-homogeneous channels; 




\item We design a new queue-based transmission protocol suited for stochastic arrivals; 

\item We numerically evaluate the stable throughput region and compare our results to the information-theoretic outer-bounds. We observe that the stable throughput region matches the known capacity region\footnote{In our numerical evaluation, we sweep different erasure probabilities in $0.02$ increments.} under specific channel conditions;

\item We quantify the computational overhead of coding (in terms of number of algebraic operations), and we analyze the lifetime (or delivery delay) under the proposed protocol.
\end{enumerate}

The rest of the paper is organized as follows. Section~\ref{Section:Problem} formulates our problem. In Section~\ref{sect:main_result}, we explain our main results through two theorems. We present our queue-based transmission protocol in Section~\ref{Section:Protocol}, and Section~\ref{sect:non-homogeneous} is dedicated to the transmission protocol in non-homogeneous settings. Section~\ref{Section:Simulation} includes our simulation results and shows by simulation that our proposed transmission protocol achieves the entire capacity region under specific cases. In Section~\ref{Section:Conclusion}, we conclude the discussion about queue-based algorithm.

\section{Problem Setting}
\label{Section:Problem}

In this section, we define the channel model and the stable throughput region. 
\subsection{Channel model}
\label{subsec:channel_model}
We consider a binary channel model for wireless networks introduced in~\cite{vahid2014capacity} which includes two transmitter-receiver pairs as in Fig.~\ref{Fig:network_model}. In this network, transmitter ${\sf Tx}_i$ only wishes to communicate with receiver ${\sf Rx}_i$, $i=1,2$. The received signal at ${\sf Rx}_i$ is given by:
\begin{equation}
\label{eq:received_signal}
Y_i[t] = S_{ii}[t] X_{i}[t] \oplus S_{\bar{i}i}[t]X_{\bar{i}}[t], \qquad i=1,2 \text{~and~} \bar{i}\overset{\triangle}=3-i.
\end{equation}

Here, $X_i[t]$ is the binary transmitted signal from ${\sf Tx}_i$ and $S_{\bar{i}i}[t]$ captures the channel gain from transmitter ${\sf Tx}_{\bar{i}}$ to receiver ${\sf Rx}_i$ at time instant $t$. We assume the channel gains are independent (across users and time) binary random variables (i.e. $S_{ij}[t] \in \{0,1\}$):
\begin{equation}
\label{eq:bernoulli_dist}
S_{ij}[t]\stackrel{i.i.d.}{\thicksim}\mathcal{B}(1-\delta_{ij}), \qquad i,j=1,2,
\end{equation}
where $0\leq \delta_{ij} \leq1$ represents the erasure probability of the link between transmitter ${\sf Tx}_{i}$ and receiver ${\sf Rx}_j$, $S_{ij}[t]$ is distributed independently from other users and across time, and ``$\oplus$'' denotes the summation in the binary field (\emph{i.e.} XOR).

In the rest of paper, we use the following notations:

\noindent i- We use uppercase letters to represent random variables and lowercase letters for variables;

\noindent ii- We introduce
\begin{equation}
   S[t]\triangleq \bigg [S_{11}[t], S_{21}[t],S_{22}[t] ,S_{12}[t] \bigg ]; 
\end{equation}

\noindent iii- We define
\begin{equation}
S_{\bar{i}i}^{\tau}\triangleq \bigg [S_{\bar{i}i}[1], S_{\bar{i}i}[2], \ldots ,S_{\bar{i}i}[\tau] \bigg ]^{\top} \text{for a natural number}~\tau; 
\end{equation}

\noindent iv- We set
\begin{align}
S_{ii}^{\tau} X_{i}^{\tau} \oplus S_{\bar{i}i}^{\tau}X_{\bar{i}}^{\tau}\triangleq \bigg [S_{ii}[1]X_{i}[1] \oplus S_{\bar{i}i}[1]X_{\bar{i}}[1],\ldots,S_{ii}[\tau]X_{i}[\tau] \oplus S_{\bar{i}i}[\tau]X_{\bar{i}}[\tau]\bigg ]^{\top}.
\end{align}


We consider a scenario in which ${\sf Tx}_i$ intends to send message $W_i \in \{1,2,\ldots, 2^{n\lambda_i}\}$ over $n$ uses of the channel, and messages are uniformly selected and mutually independent of the communication channels. 
Here, we assume the channel state information at the transmitters (CSIT) is available with a unit time delay (${\sf Tx}_1$ and ${\sf Tx}_2$ know $S_{ij}[t-1], i,j=1,2$ at time instant $t$). Moreover, we consider an encoder at ${\sf Tx}_i$ that encodes message $W_i$ to the transmitted signal, $X_i^n$, as follows:
\begin{equation}
\label{eq:encode-function}
X_i[t]=f_i(W_i,S^{t-1}), \qquad i=1,2,
\end{equation}
where $S^{t-1}=\bigg[S[1],S[2],\ldots,S[t-1]\bigg]^{\top}$ shows the delayed CSIT at time $t$ and $f_i(.)$ is the encoding function that depends on the message $W_i$ and the delayed CSIT. Furthermore, at receiver ${\sf Rx}_i$, we consider a decoder that obtains $\hat{W}_i$ as the decoded message using the decoding function $g_i$ as below:
\begin{equation}
\label{eq:decode-function}
\hat{W}_i=g_i(Y_i^n,S^n), \qquad i=1,2,
\end{equation}
In this work, we assume ${\sf Rx}_i$ knows the instantaneous CSI. Therefore, in (\ref{eq:decode-function}), $g_i$ is defined as a function of received signal $Y_i^n$ and instantaneous channel state information. An error happens when $\hat{W}_i\neq W_i$ and average probability of decoding error is:
\begin{equation}
\label{eq:error-message}
e_i^n=\mathbb{E}[\Pr\small[\hat{W}_i\neq W_i\small] ], \qquad i=1,2,
\end{equation}
where the expected value depends on the random choice of messages at both transmitters. A rate-tuple $(\lambda_1,\lambda_2)$ is achievable if there exist encoding and decoding functions at the transmitters and the receivers, respectively, that result in $e_i^n \rightarrow 0$ when $n \rightarrow \infty$. The capacity region, $\mathcal{C}$, is the closure of all achievable rate-tuples. 
\subsection{Stability definitions}


In this paper, we assume data bits arrive at ${\sf Tx}_i$ according to a Poisson $(\lambda_{i})$ distribution, $i=1,2$, $0\leq \lambda_{i} \leq 1$, and the two processes are independent from each other. We further discuss the arrival rates in Remark~\ref{Remark:SlowRateChange}.

We need the following definitions and claim to define the $\epsilon$-stable throughput region:
\begin{definition}
\label{def:decodeable-bits}
An intended bit for ${\sf Rx}_i, i=1,2$, is decodable at time instant $t$ if ${\sf Rx}_i$ either receives it without any interference or gets sufficient number of equations to recover it by time instant $t$. 
\end{definition}


\begin{definition}
\label{def:delivered_rate}
The delivered rate by time $n$, $\hat{\lambda}_i^n$, is defined as 
\begin{equation}
  \hat{\lambda}_i^n\triangleq\mathbb{E}\left[\frac{K_i^n}{n}\right], \quad i=1,2,
\end{equation}
where $K_i^n$ is the random variable denoting the number of decodable bits at ${\sf Rx}_i$ during communication time $n$.
\end{definition}

We use $Q_i^{i|\emptyset}$ to denote the initial queue into which the newly arrived bits at ${\sf Tx}_i, i=1,2,$ join. Moreover, we define $Q_i^{F}$ to describe the final queue whose bits will not be retransmitted. We introduce some intermediate queues in Section~\ref{Subsection:QI} to describe our protocol. We note each data bit in $Q_i^{i|\emptyset}$ can: 1) stay in $Q_i^{i|\emptyset}$; 2) join $Q_i^{F}$; 3) join other intermediate queues. Also, we know $\lambda_i$ and $\hat{\lambda}^n_i$ describe the arrival rate at $Q_i^{i|\emptyset}$ and $Q_i^{F}$, respectively
. As a result, 

\begin{lemma}
\label{def:delivery_rate}
For the system defined above, we have \begin{equation}
  \hat{\lambda}^n_i \leq\mathrm{average~departure~rate~of}~Q_i^{i|\emptyset}, \quad i=1,2,
\end{equation}
\end{lemma}
\begin{proof}
The proof is straightforward. We know:

i) All delivered bits at $Q_i^{F}$ come from $Q_i^{i|\emptyset}$;

ii) Some bits departed from $Q_i^{i|\emptyset}$ may not join $Q_i^{F}$ (get stuck in other queues).

As a result, $\hat{\lambda}^n_i \leq \text{average}~ \text{departure rate of}~Q_i^{i|\emptyset},~i=1,2,$ which proves the claim.
\end{proof}
Furthermore, Loynes theorem in~\cite{loynes1962stability}  states that:
\begin{theorem}
(Theorem 3 of~\cite{loynes1962stability}) A queue with $\mathbb{E}[\mathrm{arrival~rate}-\mathrm{departure~rate}]<0$ is stable if the arrival rate and departure rate of the queue are jointly stationary.
\end{theorem}
In this paper, we note that the arrival rate and departure rate of our initial queues are strictly stationary. So for the initial queue, based on Lemma~\ref{def:delivery_rate}, it suffices that the delivered rate at the corresponding receiver to be greater than 
the arrival rate at $Q_i^{i|\emptyset}$ for $i=1,2$. Based on this argument, we define the $\epsilon$-stable region as follows: 

\begin{definition}
\label{def:stable_rate}
The $\epsilon$-stable throughput region is the closure of all arrival rates $(\lambda_1,\lambda_2)$ for which
\begin{equation}
\label{eq:conditionstability}
 \lambda_{i}(1-\epsilon) < \hat{\lambda}^n_{i},\quad \textrm{and} \quad \epsilon \ll 1,  
\end{equation}
can be simultaneously satisfied, $i=1,2$, as $n \rightarrow \infty$. Here, $\epsilon$ represents the acceptable error rate, which is the percentage difference between the arrival rate at ${\sf Tx}_i$ and the delivered rate at ${\sf Rx}_i$.
\end{definition}


\begin{remark}
\label{remark:small_epsilon}
We note that if $\epsilon$ in \eqref{eq:conditionstability} goes to zero as $n \rightarrow \infty$, then Definition~\ref{def:stable_rate} is in agreement with the conventional definition of stability region. However, to make the problem tractable, for finite values of $n$ used in our simulations, we assign a small value to $\epsilon$. 
\end{remark}

In the remaining of this paper for simplicity, we drop $\epsilon$, and use the term ``stable throughput region''. When needed, we specify the value of $\epsilon$.

In this paper, we intend to compare the capacity and stable throughput regions in two-user interference network with delayed-CSIT. To do this, we derive the outer-bound of stable throughput region in non-homogeneous settings and use simulations to evaluate the results. 

\section{Main Results}
\label{sect:main_result}

We state our main results in two theorems and one claim.


\begin{theorem}
\label{thm:outer-bound_capacity}
The set of all rate-tuples $(\hat{\lambda}_1,\hat{\lambda}_2)$ satisfying
\begin{equation}
\label{eq:outer-bound}
\bar{\mathcal{C}} \equiv \left\{ \begin{array}{ll}
\vspace{1mm} 0 \leq \hat{\lambda}_i \leq (1-\delta_{ii}), & i = 1,2, \\
\hat{\lambda}_i + \frac{1-\delta_{ii}\delta_{i\bar{i}}}{(1-\delta_{i\bar{i}})}\hat{\lambda}_{\bar{i}} \leq \frac{(1-\delta_{ii}\delta_{i\bar{i}})(1-\delta_{\bar{i}\bar{i}}\delta_{i\bar{i}})}{(1-\delta_{i\bar{i}})}, & i = 1,2, \bar{i}=3-i.
\end{array} \right.
\end{equation}
includes the capacity region of the two-user intermittent interference network of Section~\ref{subsec:channel_model}.
\end{theorem}
The proof of this theorem is deferred to Appendix~\ref{apndx:A1}.

\begin{theorem}
\label{thm:C_serves_outer-bound}
The outer-bound region $\bar{\mathcal{C}}$ defined in \eqref{eq:outer-bound} also serves as an outer-bound for the stable throughput region $\mathcal{D}$ defined in Definition~\ref{def:stable_rate}.
\end{theorem}

\begin{proof}
 
We know that the outer-bound region $\bar{\mathcal{C}}$ is derived when all transmitted bits are available at $t=0$. Such a system can easily emulate stochastic arrivals: ${\sf Tx}_i, i=1,2,$ generates a random variable with Poisson $(\lambda_{i})$ distribution and selects from the available bits a subset of size equal to the outcome of the random variable at each time. In other words, any stable rate-tuple must be an achievable rate-tuple. This implies that ${\mathcal{D}}\subset\bar{\mathcal{C}}$ and proves the theorem.
\end{proof}

In this work, we use simulation to evaluate our results for finite values of $n$ and small values of $\epsilon$. Therefore, we define $\bar{\mathcal{C}}_{\epsilon}$ as the simulated outer-bound region as below:
\begin{equation}
  \bar{\mathcal{C}}_{\epsilon}\triangleq \{(\hat{\lambda}_1,\hat{\lambda}_2)|\hat{\lambda}_1,\hat{\lambda}_2\ge 0, (\hat{\lambda}_1+\epsilon,\hat{\lambda}_2+\epsilon)\in \bar{\mathcal{C}}\}.  
\end{equation}

\begin{claim}
\label{claim:achieves_different_case_simulation}
We propose a linear queue-based transmission protocol that we evaluate through simulations and show $\mathcal{D}\equiv\bar{\mathcal{C}}_{\epsilon}$ for $\epsilon=0.01$ under the following cases: 

\noindent 1- Homogeneous channel settings ($\delta_{ij}=\delta, i,j=1,2$) when $0 \leq \delta \leq 1$;

\noindent 2- Symmetric channel settings ($\delta_d=\delta_{ii}, \delta_c=\delta_{i\bar{i}},i = 1,2, \bar{i}=3-i$) when $\delta_d\leq 0.5$ and $0\leq \delta_c\leq \frac{1}{2-\delta_d}$;

\noindent 3- Non-homogeneous settings when $\frac{3-\sqrt{5}}{2}<\delta_{ii}\leq 0.5$, $0 \leq \delta_{\bar{i}\bar{i}}\leq \frac{3-\sqrt{5}}{2}$, $0.5 \leq \delta_{i\bar{i}}< \frac{1}{2-\delta_{ii}}$, $\delta_{ii} < \delta_{\bar{i}i} \leq 0.5$, $,i = 1,2$, and $\bar{i}=3-i$.

\noindent 4- 
For symmetric channel settings when $\frac{3-\sqrt{5}}{2}\leq \delta_d\leq 0.5$ and $\frac{1}{2-\delta_d} < \delta_c\leq 1$, we numerically show the inner and outer bounds gap.
\end{claim}

Fig.~\ref{Fig:corner} depicts the stable throughput region with delayed-CSIT when $\delta_{11}=0.4$, $\delta_{12}=0.6$, $\delta_{22}=0.2$, and $\delta_{21}=0.5$. We will explain the details of our transmission protocol in Sections~\ref{Section:Protocol} and \ref{sect:non-homogeneous} and show $\mathcal{D}\equiv{\mathcal{C}}_{\epsilon}$ in the above cases by simulation in Section~\ref{Section:Simulation}. 

In this work, we define cornerpoint ``A'' as the point where $\hat{\lambda}_1=1-\delta_{11}$ and $\hat{\lambda}_2$ is maximum. Similarly, we define cornerpoint ``B'' by changing the users' labels. Moreover, we define cornerpoint ``C'' as the sum-rate cornerpoint where $\hat{\lambda}_1+\hat{\lambda}_2$ is maximum, see Fig.~\ref{Fig:corner}.

\begin{figure}[!ht]
  \centering

\tikzset{every picture/.style={line width=0.6pt}} 

\begin{tikzpicture}[x=0.6pt,y=0.6pt,yscale=-.85,xscale=0.85]

\draw [line width=1.5]  (67.5,260.66) -- (289.64,260.66)(89.71,61.25) -- (89.71,282.82) (282.64,255.66) -- (289.64,260.66) -- (282.64,265.66) (84.71,68.25) -- (89.71,61.25) -- (94.71,68.25)  ;
\draw [color={rgb, 255:red, 208; green, 2; blue, 27 }  ,draw opacity=1 ][line width=1.5]  [dash pattern={on 1.69pt off 2.76pt}]  (250.52,107.9) -- (250.52,260.29) ;
\draw [color={rgb, 255:red, 208; green, 2; blue, 27 }  ,draw opacity=1 ][line width=1.5]  [dash pattern={on 1.69pt off 2.76pt}]  (90.14,141.13) -- (280.5,140.43) ;
\draw [color={rgb, 255:red, 208; green, 2; blue, 27 }  ,draw opacity=1 ][line width=1.5]  [dash pattern={on 1.69pt off 2.76pt}]  (122.5,126.43) -- (272.5,195.43) ;
\draw [color={rgb, 255:red, 208; green, 2; blue, 27 }  ,draw opacity=1 ][line width=1.5]  [dash pattern={on 1.69pt off 2.76pt}]  (194.5,118.43) -- (261.5,250.43) ;
\draw [color={rgb, 255:red, 74; green, 88; blue, 226 }  ,draw opacity=1 ][line width=1.5]    (250.5,229.43) -- (250.52,260.29) ;
\draw [color={rgb, 255:red, 74; green, 88; blue, 226 }  ,draw opacity=1 ][line width=1.5]    (153.19,140.9) -- (221.5,171.43) ;
\draw [color={rgb, 255:red, 74; green, 88; blue, 226 }  ,draw opacity=1 ][line width=1.5]    (221.5,171.43) -- (250.5,229.43) ;
\draw [color={rgb, 255:red, 74; green, 88; blue, 226 }  ,draw opacity=1 ][fill={rgb, 255:red, 74; green, 144; blue, 226 }  ,fill opacity=1 ][line width=1.5]    (90.14,141.13) -- (153.05,141.13) ;
\draw [color={rgb, 255:red, 208; green, 2; blue, 27 }  ,draw opacity=1 ][line width=1.5]  [dash pattern={on 1.69pt off 2.76pt}]  (261.57,73.29) -- (296.57,73.29) ;
\draw [color={rgb, 255:red, 74; green, 88; blue, 226 }  ,draw opacity=1 ][fill={rgb, 255:red, 74; green, 144; blue, 226 }  ,fill opacity=1 ][line width=1.5]    (260.57,98.29) -- (296,98.29) ;

\draw (58,134) node [anchor=north west][inner sep=0.75pt]  [font=\footnotesize]  {${\textstyle 0.6}$};
\draw (242,266) node [anchor=north west][inner sep=0.75pt]  [font=\footnotesize]  {$0.8$};
\draw (150.2,123.17) node [anchor=north west][inner sep=0.75pt]  [font=\footnotesize] [align=left] {\textbf{A}};
\draw (221.33,156.24) node [anchor=north west][inner sep=0.75pt]  [font=\scriptsize] [align=left] {\textbf{C}};
\draw (254.33,219.24) node [anchor=north west][inner sep=0.75pt]  [font=\scriptsize] [align=left] {\textbf{B}};
\draw (62,52) node [anchor=north west][inner sep=0.75pt]    {$\hat{\lambda}_{1}$};
\draw (278.67,268) node [anchor=north west][inner sep=0.75pt]    {$\hat{\lambda}_{2}$};
\draw (299,65) node [anchor=north west][inner sep=0.75pt]   [align=left] {{\footnotesize Information-theoritic boundaries}};
\draw (299,89) node [anchor=north west][inner sep=0.75pt]   [align=left] {{\footnotesize Stable throughput region}};
\draw (115,145) node [anchor=north west][inner sep=0.7pt]  [font=\normalsize]  {$\binom{0.6}{0.36}$};
\draw (158,172) node [anchor=north west][inner sep=0.7pt]  [font=\normalsize]  {$\binom{0.4397}{0.6486}$};
\draw (190,225) node [anchor=north west][inner sep=0.7pt]  [font=\normalsize]  {$\binom{0.152}{0.8}$};

\end{tikzpicture}
  \caption{\it Stable throughput region with Delayed-CSIT when $\delta_{11}=0.4$, $\delta_{12}=0.6$, $\delta_{22}=0.2$ and $\delta_{21}=0.5$. \label{Fig:corner}}
\end{figure}
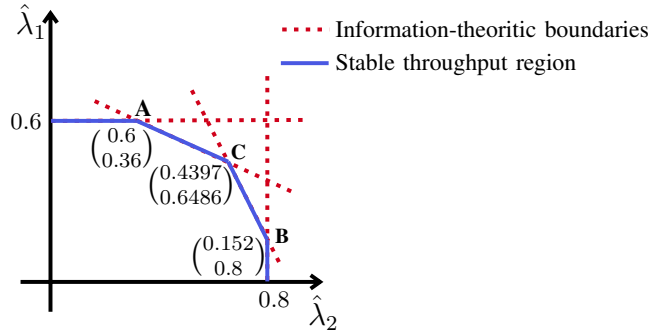
\section{Transmission Protocol}
\label{Section:Protocol}

In this section, we propose a general transmission protocol which exploits various queues and coding opportunities to achieve the entire stable throughput region. 
We focus on homogeneous settings where $\delta_{ij}=\delta$ and divide $\delta$ into three parts as $0\leq \delta < \frac{3-\sqrt{5}}{2}$, $\frac{3-\sqrt{5}}{2}\leq \delta < 0.5$ and $0.5\leq \delta \leq 1$. In this section, we explain our transmission protocol for $0\leq \delta < \frac{3-\sqrt{5}}{2}$ and defer the detailed description of our protocol for other regimes to appendices. Furthermore, we demonstrate the protocol through the following: (1) Identifying a priority policy to distinguish which queue should send its data bits before others; (2) Defining the update procedure to employ priority policy at each node; (3) Defining control table to determine data movement rules between queues; (4) Introducing some flexible cases which provide a general rule to achieve the entire stable throughput region; (5) Proposing a pseudo-code to explain how to implement our general protocol. The details of our protocol when $\frac{3-\sqrt{5}}{2}\leq \delta < 0.5$, $0.5\leq \delta \leq 1$ and non-homogeneous settings are presented in Appendices~\ref{apndx:A},~\ref{apndx:B}, and Section~\ref{sect:non-homogeneous}, respectively. We note that all the operations in our transmission protocol are linear.
\subsection{Constructing virtual Queues at the transmitters}
\label{Subsection:QI}
In queue-based transmissions, queues play the main role to handle the data bits. In each queue, we use a subscript $i$ to describe that the queue stores the data bits which belong to ${\sf Tx}_i$. In addition, all queues with subscript $i$ are only available at ${\sf Tx}_i$ (i.e. there is no data exchanging between transmitters). Below, we introduce some queues for data controlling:

$Q_{i}^{i|\emptyset}$ (Initial queue): It is the first queue where each bit of ${\sf Tx}_i$ goes in. It means the newly arrived bits at each time instant must be saved in $Q_{i}^{i|\emptyset}$;

$Q_{i}^{1,2|\emptyset}$(Common-interest bits): It contains data bits which belong to ${\sf Tx}_i$ and are helpful to both ${\sf Rx}_1$ and ${\sf Rx}_2$;

$Q_{i}^{i|\bar{i}}$(Side-information bits at the unintended receiver): It contains the data bits that belong to ${\sf Tx}_i$, needed by ${\sf Rx}_i$, but have become available to ${\sf Rx}_{\bar{i}}$ without any interference. Thus, ${\sf Rx}_{\bar{i}}$ can use these bits as side information in future process;

$Q_{i}^{\bar{i}|i}$(Side-information bits at the intended receiver): It contains the bits that belong to ${\sf Tx}_{i}$ and are helpful to ${\sf Rx}_{\bar{i}}$ even though they are available at ${\sf Rx}_i$ without any interference. Hence, ${\sf Rx}_i$ will be able to employ them as side information;

$Q_{i}^{c_{1}}$ (Special queue): If at the time of transmission, the bit of $\sf{Tx}_i$ is either in $Q_{i}^{i|\emptyset}$ or $Q_{i}^{1,2|\emptyset}$, $i=1,2$, and if all links are on (see Table~\ref{table:flexible queues} as an example), then the bit will join $Q_{i}^{c_{1}}$.


$Q_{i}^{F}$: This queue contains the data bit that belongs to ${\sf Tx}_i$ and satisfies the following conditions:
(1) ${\sf Rx}_i$ obtains this data bit without interference or receives sufficient number of linearly independent equations to decode it, and (2) the data bit is not considered to be helpful to ${\sf Rx}_{\bar{i}}$. We note 
the bits of $Q_{i}^{F}$ will not be retransmitted.

When needed, we will explain when transmitted bits may remain in their respective queues (\emph{i.e.} no status change).

\subsection{Delivering Options}
\label{Subsection:DO}
In this part, we identify two delivering opportunities based on prior values of channel gain $S_{ij}(i,j=1,2)$. Later, we will use these opportunities to deliver data bits to the intended destinations by exploiting the side information created due to channel statistics.

\textbf{Creating common-interest bits:}
In this category, ${\sf Tx}_i$ exploits the delayed CSIT to deliver its data bits to ${\sf Rx}_i$. We use Fig.~\ref{Fig:intend}(a) as an example to clarify this point. In Fig.~\ref{Fig:intend}(a), ${\sf Tx}_1$ and ${\sf Tx}_2$ send out $a_{1}$ and $b_{1}$, respectively, at the same time to intended receivers. We assume all channels are on at time $t$ and both receivers obtain $a_{1}\oplus b_{1}$. We observe that if ${\sf Tx}_1$ delivers $a_{1}$ to both receivers (in the future), ${\sf Rx}_1$ attains $a_{1}$, and ${\sf Rx}_2$ removes $a_{1}$ from $a_{1}\oplus b_{1}$ and decodes $b_{1}$ easily. As a result, $a_{1}$ is considered as a common interest bit and joins $Q_{1}^{1,2|\emptyset}$ and it is not needed to retransmit $b_{1}$ and $b_{1}$ joins $Q_{2}^{F}$. In this case,~$b_{1}$ is considered as a virtually delivered bit which means $b_{1}$ can be decoded by ${\sf Rx}_2$ if $a_{1}$ is delivered to ${\sf Rx}_2$.

\begin{figure}[h!]
  \centering
  \includegraphics[trim = 0mm 0mm 0mm 0mm, clip, scale=3.8, width=0.48\linewidth]{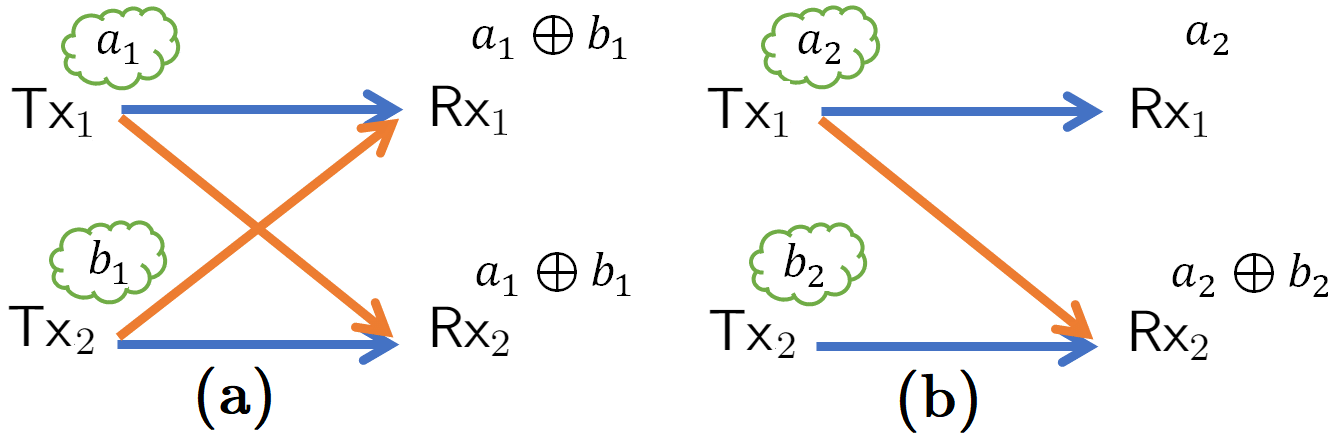}
  \caption{\it 
  (a) ${\sf Rx}_1$ and ${\sf Rx}_2$ receive $a_{1} \oplus b_{1}$ at time $t$. We consider $a_{1}$ as a common interest bit and move it to $Q_{1}^{1,2|\emptyset}$ and move $b_{1}$ to $Q_{2}^{F}$ as a virtually delivered bit. (b) ${\sf Rx}_1$ gets $a_{2}$ directly from ${\sf Tx}_1$ and ${\sf Rx}_2$ obtains $a_{2} \oplus b_{2}$. In this case, the optimal choice is to retransmit $a_{2}$ and consider $b_{2}$ as a virtually delivered bit in $Q_{2}^{F}$.} 
  \label{Fig:intend}
  \vspace{-4mm}
\end{figure}

\textbf{Exploiting available side information:}
In this category, we describe how ${\sf Tx}_i$ can exploit the available side information at ${\sf Rx}_{\bar{i}}$ to deliver data bits more efficiently. Assume as an example the scenario of Fig.~\ref{Fig:intend} (b) where $a_{2}$ and $b_{2}$ are transmitted by ${\sf Tx}_1$ and ${\sf Tx}_2$, respectively. In this case, ${\sf Rx}_1$ gets $a_{2}$ and decodes it successfully, and ${\sf Rx}_2$ obtains $a_{2} \oplus b_{2}$. We observe that there are two options for transmitters to deliver $b_{2}$ to ${\sf Rx}_2$ : 1) ${\sf Tx}_2$ retransmits $b_{2}$ to ${\sf Rx}_2$ which can create additional interference at ${\sf Rx}_1$; 2) ${\sf Tx}_1$ retransmits $a_{2}$ to ${\sf Rx}_2$ which provides a more efficient solution as it does not interfere with ${\sf Rx}_1$. Therefore, ${\sf Tx}_1$ can help ${\sf Rx}_2$ by transmitting $a_{2}$ ($a_{2}$ is saved in $Q_{1}^{2|1}$) and $b_{2}$ is delivered virtually since ${\sf Tx}_2$ does not retransmit it.



\begin{remark}
\label{remark:fault_example}
To illustrate the main error in the communication protocol of~\cite{pan2016stability}, we consider the example illustrated in Fig.~\ref{Fig:fault_example} which happens during four time instants $t_1$, $t_2$, $t_3$ and $t_4$. Suppose ${\sf Tx}_1$ and ${\sf Tx}_2$ send out $a_{1}$ and $b_{1}$, respectively, and only $S_{11}[t]$ and $S_{21}[t]$ are on at time instant $t_1$ (i.e. ${\sf Rx}_1$ receives $a_{1}\oplus b_{1}$). Similarly, assume ${\sf Tx}_1$ and ${\sf Tx}_2$ send out $a_{2}$ and $b_{2}$, respectively, and $S_{22}[t]$ and $S_{12}[t]$ are on at time instant $t_2$. At time instant $t_3$, consider ${\sf Tx}_1$ and ${\sf Tx}_2$ transmit $a_{3}$ and $b_{3}$, respectively, and $S_{11}[t]$, $S_{12}[t]$, and $S_{22}[t]$ are on. Finally, assume ${\sf Tx}_1$ and ${\sf Tx}_2$ send out $a_{4}$ and $b_{4}$, respectively, and $S_{11}[t]$, $S_{21}[t]$ and $S_{22}[t]$ are on at time instant $t_4$. According to~\cite{pan2016stability}'s approach, $a_{1}\rightarrow Q_{1}^{F}$, $b_{1}\rightarrow Q_{2}^{2|1}$, $a_{2}\rightarrow Q_{1}^{1|2}$, $b_{2}\rightarrow Q_{2}^{F}$, $a_{3}\rightarrow Q_{1}^{2|1}$, $b_{3}\rightarrow Q_{2}^{F}$, $a_{4}\rightarrow Q_{1}^{F}$ and $b_{4}\rightarrow Q_{2}^{1|2}$ and it is sufficient to deliver $b_{1}\oplus b_{4}$ and $a_{2}\oplus a_{3}$ to both receivers to recover \{$a_{1},a_{2},a_{3},a_{4}$\} and \{$b_{1},b_{2},b_{3},b_{4}$\} by ${\sf Rx}_1$ and ${\sf Rx}_2$, respectively. However, we observe proposed decoding method does not work properly. For example, if ${\sf Rx}_1$ receives $b_{1}\oplus b_{4}$ and $a_{2}\oplus a_{3}$, it will decode $a_{2}$ from $a_{2}\oplus a_{3}$ but there is no way to decode $a_{1}$ and $a_{4}$ for ${\sf Rx}_1$. A similar statement holds for ${\sf Rx}_2$ which cannot decode $b_{2}$ and $b_{3}$. According to our strategy which we will discuss shortly, transmitted bits should join the following queues: $a_{1}\rightarrow Q_{1}^{F}$, $b_{1}\rightarrow Q_{2}^{1,2|\emptyset}$, $a_{2}\rightarrow Q_{1}^{1,2|\emptyset}$, $b_{2}\rightarrow Q_{2}^{F}$,  $a_{3}\rightarrow Q_{1}^{2|1}$, $b_{3}\rightarrow Q_{2}^{F}$, $a_{4}\rightarrow Q_{1}^{F}$, and $b_{4}\rightarrow Q_{2}^{1|2}$. In this case, it is enough to deliver $b_{1}$ and $a_{2}$ to both receivers,  $a_{3}$ to ${\sf Rx}_2$, and $b_{4}$ to ${\sf Rx}_1$.
\end{remark}
\begin{figure}[h!]
  \centering
  \includegraphics[trim = 0mm 0mm 0mm 0mm, clip, scale=5.5, width=0.9\linewidth]{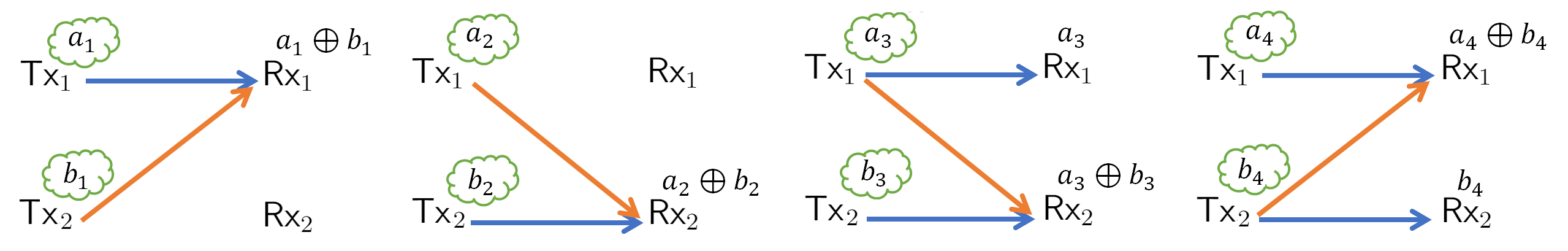}
  \caption{\it An example of decoding problem in~\cite{pan2016stability} that claims delivering $b_{1}\oplus b_{4}$ and $a_{2}\oplus a_{3}$ to both receivers is enough to decode the desired bits. However, ${\sf Rx}_1$ cannot decode $a_{1}$ and $a_{4}$ for example.\label{Fig:fault_example}}
  \vspace{-5mm}
\end{figure}

We denote the status of bits at the beginning and the end of each time instant by origin and destination queues, respectively. For simplicity, we use the decimal representation of the four-bit $S[t]$, which we refer to as the Situation Number (SN) as in Table~\ref{table:SN_numbers}, and use two queues next to each other separated by commas to indicate the queues (origin or destination) at both transmitters.

\begin{table}[ht]
\caption{Map $S[t]$ to Situation Number (SN)}
\centering
\begin{tabular}{|c|c|c|c|c|c|c|c|}
\hline
SN & $S[t]$ & SN & $S[t]$ & SN & $S[t]$ & SN & $S[t]$ \\[0.5ex]

\hline \hline
$1$ & $[1,1,1,1]$ & $5$ & $[1,0,0,0]$ & $9$ & $[0,0,1,0]$ & $13$ & $[0,1,0,0]$ \\
\hline
$2$ & $[1,0,1,1]$ & $6$ & $[1,0,0,1]$ & $10$ & $[0,1,1,0]$ & $14$ & $[0,0,0,1]$ \\
\hline
$3$ & $[1,1,1,0]$ & $7$ & $[1,1,0,0]$ & $11$ & $[0,0,1,1]$ & $15$  & $[0,1,0,1]$ \\
\hline
$4$ & $[1,0,1,0]$ & $8$ & $[1,1,0,1]$ & $12$ & $[0,1,1,1]$ & $16$ & $[0,0,0,0]$ \\[1ex]  
\hline
\end{tabular}
\label{table:SN_numbers}
\vspace{-8mm}
\end{table}
\subsection{XOR Opportunities}
\label{Subsection:XOR_COM}
In this part, we discuss XOR opportunities to construct common interest bits based on the constraints explained in Section~\ref{Subsection:DO}. We note that the data bits leave their origin queues immediately after XOR combination happens. We show that three XOR options exist as below:

$\bullet$ \textbf{XOR between $\mathbf{Q_{i}^{c_{1}}}$ and $\mathbf{Q_{i}^{i|\bar{i}}}$}:

Suppose $a_{1}$ and $b_{1}$ are available in $Q_1^{1,2|\emptyset},Q_2^{2|\emptyset}$ and SN-1 happens and in another time instant, $a_{2}$ and $b_{2}$ are the available bits in $Q_1^{1|\emptyset},Q_2^{2|\emptyset}$ and SN-15 occurs (as Fig.~\ref{Fig:model_1}). According to the definition of $Q_i^{c_1}$, we consider $a_1$ and $b_1$ as communication bits. Also, $a_2$ and $b_2$ are received by unintended receivers and ${\sf Rx}_{\bar{i}}$ can consider them as side information ($a_{2}\rightarrow Q_{1}^{1|2}$ and $b_{2}\rightarrow Q_{2}^{2|1}$). In this part, we show that delivering $a_{1}\oplus a_{2}$ and $b_{1}\oplus b_{2}$ to both receivers leads to deliver $\{a_{1},a_{2}\}$ and $\{a_{1},b_{1},b_{2}\}$ to ${\sf Rx}_1$ and ${\sf Rx}_2$, respectively. We describe the details of decoding procedure at ${\sf Rx}_1$ as an example. If $a_{1}\oplus a_{2}$ and $b_{1}\oplus b_{2}$ are delivered to ${\sf Rx}_1$, it uses $b_{2}$ and $b_{1}\oplus b_{2}$ to decode $b_{1}$ and then employs $b_{1}$ to get $a_{1}$ from $a_{1}\oplus b_{1}$ and uses $a_{1}$ to decode $a_{2}$ from $a_{1}\oplus a_{2}$. Similarly, ${\sf Rx}_2$ decodes its desired data bits. Hence, combining $Q_{i}^{c_{1}}$ and $Q_{i}^{i|\bar{i}}$ is a XOR model in our network. At this time, ${\sf Tx}_i$ omits the bits from $Q_{i}^{c_{1}}$ and $Q_{i}^{i|\bar{i}}$ and puts XORed bit in $Q_{i}^{1,2|\emptyset}, i\in \{1,2\}$.

\begin{figure}[h!]
  \centering
  \includegraphics[trim = 0mm 0mm 0mm 0mm, clip, scale=4, width=0.45\linewidth]{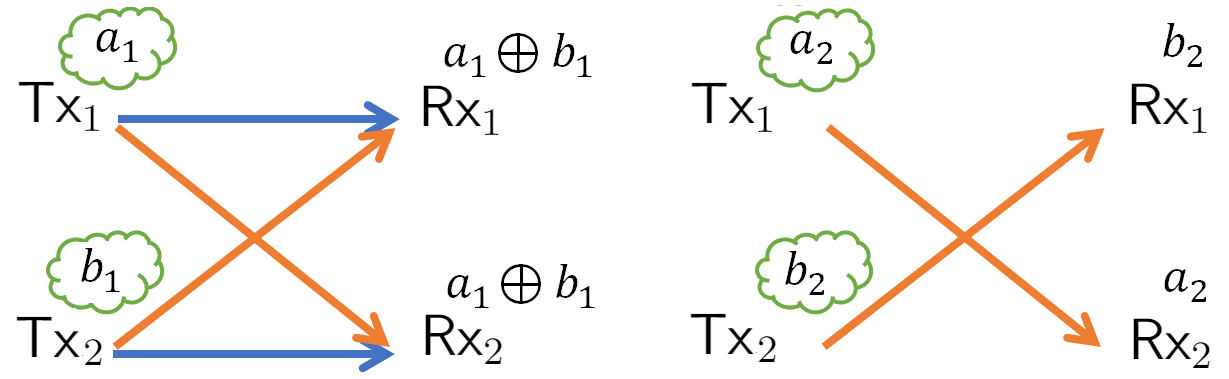}
  \caption{\it Consider SN-1 and SN-15 happen at two different time instants when $Q_1^{1,2|\emptyset},Q_2^{2|\emptyset}$ and $Q_1^{1|\emptyset},Q_2^{2|\emptyset}$ are the origin queues, respectively. It is enough to deliver $a_{1}\oplus a_{2}$ and $b_{1}\oplus b_{2}$ to both receivers. Thus, combining $Q_{i}^{c_{1}}$ and $Q_{i}^{i|\bar{i}}$ is a XOR opportunity.  \label{Fig:model_1}}
\end{figure}

$\bullet$ \textbf{XOR between $\mathbf{Q_{i}^{c_{1}}}$ and $\mathbf{Q_{i}^{\bar{i}|i}}$}:

Assume $a_{1}$ and $b_{1}$ are available in $Q_1^{1,2|\emptyset},Q_2^{2|\emptyset}$ and SN-1 happens and in two other time instants, $\{a_{2}, b_{2}\}$ and $\{a_{3}, b_{3}\}$ are the available bits in $Q_1^{1|\emptyset},Q_2^{2|\emptyset}$ and SN-2 and SN-13 occur, respectively. Based on delivering options in Section~\ref{Subsection:DO} and definition of $Q_i^{c_1}$, we find that $a_{1}\rightarrow Q_{1}^{c_1}$, $b_{1}\rightarrow Q_{2}^{c_1}$, $a_{2}\rightarrow Q_{1}^{2|1}$, $b_{2}\rightarrow Q_{2}^{F}$ $a_{3}\rightarrow Q_{1}^{F}$, and $b_{3}\rightarrow Q_{2}^{1|2}$. We indicate that it is enough to deliver $a_{1}\oplus a_{2}$ and $b_{1}\oplus b_{3}$ to both receivers. At this time, we describe decoding procedure for ${\sf Rx}_2$ as an instance and it is similar for ${\sf Rx}_1$. In Fig.~\ref{Fig:model_2}, ${\sf Rx}_2$ recovers $b_{1}$ by removing $b_{3}$ from $b_{1}\oplus b_{3}$ and then uses $b_{1}$ to decode $a_{1}$ from linear combination in SN-1. Furthermore, ${\sf Rx}_2$ obtains $a_{2}$ from $a_{1}\oplus a_{2}$ and uses $a_{2}$ to decode $b_{2}$ from $a_{2}\oplus b_{2}$. As a result, $Q_{i}^{c_{1}}$ and $Q_{i}^{\bar{i}|i}$ is a XOR opportunity and bits leave $Q_{i}^{c_{1}}$ and $Q_{i}^{\bar{i}|i}$ and XORed bits move to $Q_{i}^{1,2|\emptyset}, i\in \{1,2\}$.
\begin{figure}[h!]
  \centering
  \includegraphics[trim = 0mm 0mm 0mm 0mm, clip, scale=9.5, width=0.7\linewidth]{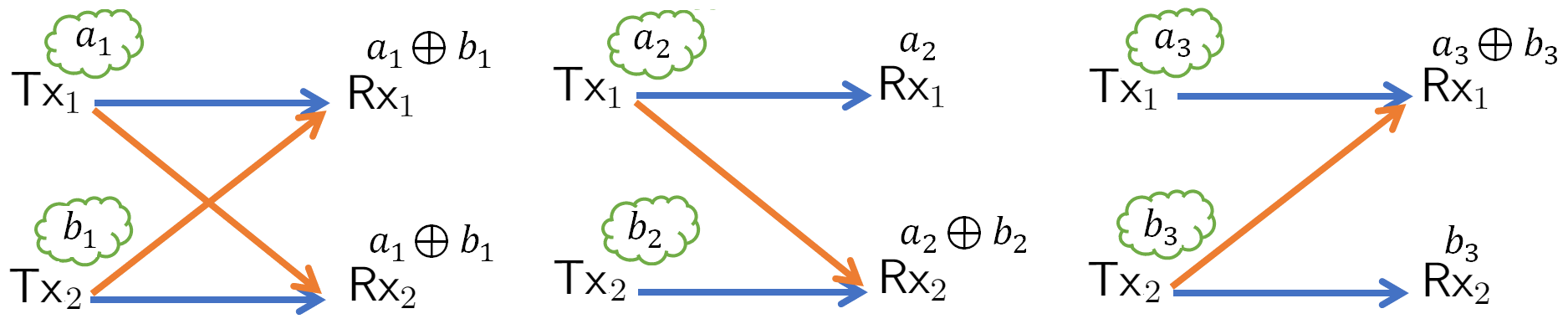}
  \caption{\it Suppose SN-1, SN-2 and SN-3 occur at three different time instants. We show that by delivering $a_{1}\oplus a_{2}$ and $b_{1}\oplus b_{3}$ to both receivers, ${\sf Rx}_i, i\in \{1,2\}$ can decode its desired bits. Therefore, $Q_{i}^{c_{1}}$ and $Q_{i}^{\bar{i}|i}$ create a XOR opportunity. \label{Fig:model_2}}
\end{figure}

$\bullet$ \textbf{XOR between $\mathbf{Q_{i}^{i|\bar{i}}}$ and $\mathbf{Q_{i}^{\bar{i}|i}}$}:

We describe this XOR combination through an example of decoding procedure for ${\sf Rx}_1$. Imagine SN-2 and SN-14 happen at two time instants when $Q_1^{1|\emptyset},Q_2^{2|\emptyset}$ are the origin queues (as in Fig.~\ref{Fig:model_3}). According to delivering option in Section~\ref{Subsection:DO} and definition of $Q_i^{i|\bar{i}}$, we find that $a_{1}\rightarrow Q_{1}^{2|1}$, $b_{1}\rightarrow Q_{2}^{F}$, $a_{2}\rightarrow Q_{1}^{1|2}$, and $b_{2}\rightarrow Q_{2}^{2|\emptyset}$. We show that by delivering $a_{1}\oplus a_{2}$ to both receivers, ${\sf Rx}_1$ obtains $a_{2}$ by removing $a_{1}$ from $a_{1}\oplus a_{2}$ and ${\sf Rx}_2$ obtains $a_{1}$ from $a_{1}\oplus a_{2}$ and uses $a_{1}$ to decode $b_{1}$ from  $a_{1}\oplus b_{1}$. Thus, $Q_{i}^{i|\bar{i}}$ and $Q_{i}^{\bar{i}|i}$ can be considered as a XOR combination, and $a_{1}$ and $a_{2}$ leave $Q_{1}^{1|2}$ and $Q_{1}^{2|1}$, and XOR value of them joins $Q_{1}^{1,2|\emptyset}$. Similar statement holds for ${\sf Rx}_2$ when SN-3 and SN-13 happen at two time instants.
\begin{figure}[h!]
  \centering
  \includegraphics[trim = 0mm 0mm 0mm 0mm, clip, scale=4.2, width=0.42\linewidth]{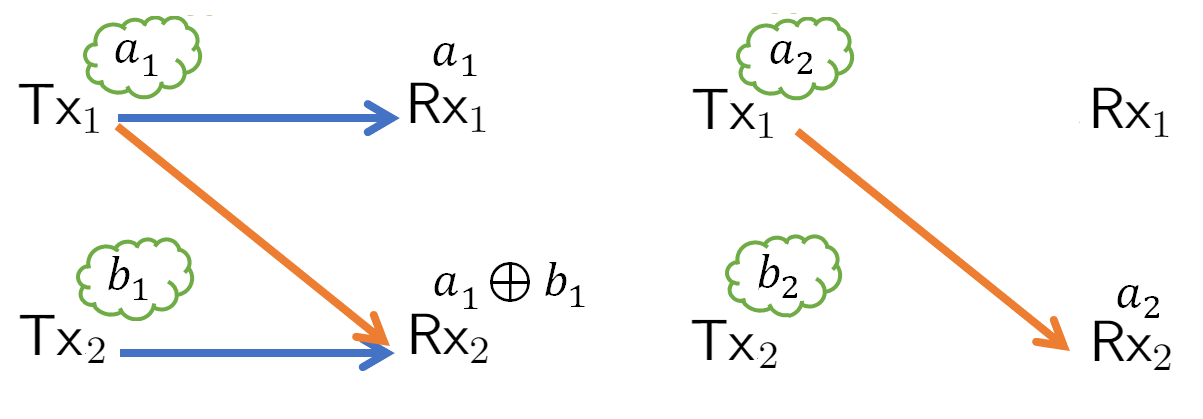}
  \caption{\it Assume SN-2 and SN-14 happen at two time instants. In this case, by delivering $a_{1}\oplus a_{2}$ to both receivers, ${\sf Rx}_1$ obtains $a_{1}$ and $a_{2}$ and it makes a XOR combination between $Q_{i}^{i|\bar{i}}$ and $Q_{i}^{\bar{i}|i}$. \label{Fig:model_3}}
\end{figure}

To use the benefits of delivering options and XOR combinations, the transmitted bits should be tracked at each node. We will describe how each node can track the transmitted bits without any confusion in Section~\ref{subsect:Transparent_CT}.
\subsection{Priority Policy}
\label{Subsection:P}

In non-stochastic queue-based transmissions, all data bits are available in $Q_{i}^{i|\emptyset}$ at time instant $t=0$ and transmission strategy is divided into two phases~\cite{vahid2014capacity}. During the first phase, transmitters send their data bits and find some combining opportunities between data bits in the second phase. However, in stochastic arrivals, the data bits are arriving during communication time. As we discussed in Section~\ref{Subsection:QI}, we have six different virtual queues at each transmitter which result in 36 different combinations as origin queues in our network. We know that one of them is $Q_1^F,Q_2^F$ which cannot be considered as the origin queues because the bits in these queues are already delivered. Therefore, we have 35 different active origin queues and it is essential to determine which data bits have higher priority to send at each time instant. Hence, we assign a priority policy (see Table~\ref{table:priority}) to determine which bits should be sent before others. Since common interest bits include more information rather than other bits (i.e. they help both receivers to decode their desired bits) $Q_{i}^{1,2|\emptyset}$ has highest priority. $Q_{i}^{i|\emptyset}$ contains new arrived bits at ${\sf Tx}_i$ and occupies second place. Usually, ${\sf Tx}_i$ should deliver desired bits to ${\sf Rx}_i$ and some of the times, it is beneficial to deliver the interfering bit to ${\sf Rx}_{\bar{i}}$. Therefore, we assign third and fourth priority to $Q_{i}^{i|\bar{i}}$ and $Q_{i}^{\bar{i}|i}$, respectively. Next place is occupied by $Q_{i}^{c_{1}}$ which contains the bits which provide XOR options for transmitters. The main role of $Q_{i}^{c_{1}}$ is providing XOR opportunities 
that simultaneously deliver data bits to both receivers.
\begin{table}[ht]
\caption{priority policy indicates which bits should leave their origin queues at time instant $t$}
\centering
\begin{tabular}{|c|c|l|}
\hline
Priority & Queue & \hspace{1.35in} Policy \\[0.5ex]

\hline \hline
1 & $Q_{i}^{1,2|\emptyset}$ & If there is a bit in $Q_{i}^{1,2|\emptyset}$  \\
2 & $Q_{i}^{i|\emptyset}$ & If $Q_{i}^{1,2|\emptyset}$ is empty and there is a bit in $Q_{i}^{i|\emptyset}$  \\
3 & $Q_{i}^{i|\bar{i}}$ & If $Q_{i}^{1,2|\emptyset}$ and $Q_{i}^{i|\emptyset}$ are empty and there is a bit in $Q_{i}^{i|\bar{i}}$  \\
4 & $Q_{i}^{\bar{i}|i}$ & If $Q_{i}^{1,2|\emptyset}$, $Q_{i}^{i|\emptyset}$ and $Q_{i}^{i|\bar{i}}$ are empty and there is a bit in $Q_{i}^{\bar{i}|i}$  \\
5 & $Q_{i}^{c_{1}}$ & If all other queues are empty and only there is a bit in $Q_{i}^{c_{1}}$ \\[1ex]  
\hline
\end{tabular}
\label{table:priority}
\end{table}
We note that each node needs the number of data bits in all queues to use Table~\ref{table:priority}. However, only ${\sf Tx}_i$ knows the number of available bits in $Q_i^{i|\emptyset}$. Thus, in next section, we define update procedure to help all nodes learn the number of available bits in initial queues.
\subsection{Update}
\label{subsect:update}

The bits arrived during the interval between two consecutive updates will not be transmitted until the next update happens. The new update will enable all other nodes to follow the status of the bits correctly. Each transmitter runs its update every $L$ time instants. For reasonably large values of $L$, using concentration results, we can show the number of newly arrived bits at transmitter $i$ is close to the expected value, $\lambda_i L$. Since $\lambda_i \leq 1$, the overhead associated with the updates is upper bounded by $\log(L)/L$ bits per channel use. Further, we note that each channel use accommodates a bit in the forward channel that represents a data packet of potentially thousands of bytes. In essence, the overhead can be made arbitrarily small and will not affect the rates. We also note that $L$ would be a modest finite number and the asymptotic information-theoretic schemes, such as the one in~\cite{vahid2014capacity}, cannot be used in this case. In the description of the transmission protocol, we ignore the overhead associated with the updates.

Also, as we will describe later in Section~\ref{subsect:find_stability_zones}, ${\sf Tx}_i, i=1,2$ notifies other nodes of changes in its arrival rate during the update procedure. Similar to the above, we can show that the overhead associated with this notification is negligible.

\subsection{Control Table}
\label{Subsection:CT}
Stochastic data arrival introduces new challenges when compared to the non-stochastic model as we discussed. In this section, we define a control table to describe how data bits move from the origin to the destination queues when $0\leq \delta < \frac{3-\sqrt{5}}{2}$ based on different SNs. There are 35 possible combinations of the origin queues and 16 SNs. Hence, the control table would be of size $35\times16$ (available online at \cite{key}) and describes the movement rules. We will describe the details of the control table when $Q_1^{1,2|\emptyset},Q_2^{2|\emptyset}$ are the origin queues as an example in Appendix~\ref{apndx:ex_CT}. Here, we propose a general rule: (1) if there is only one choice for destination queues, then that option will be chosen, and (2) when more than one choice is available, the rule selects the destination queues (called destination queues that provide more XOR opportunities) that result in a higher number of bits available to create XOR combinations described in Section~\ref{Subsection:XOR_COM}. Further, there are cases according to Table~\ref{table:flexible queues} that involve flexible choice of destination queues.
\begin{table}[hb]
\caption{Flexible destination queues when $0 \leq \delta \leq \frac{3-\sqrt{5}}{2}$ and each bit is at the origin queue of its transmitter.}

\centering
\begin{tabular}{|c|c|c|c|c|c|c|c|c|c|}
\hline
Flexible ID & SN & $P[t]=A$& $P[t]=C$& $P[t]=B$ & Flexible ID & SN & $P[t]=A$& $P[t]=C$& $P[t]=B$\\[0.5ex]

\hline \hline
$F_1$ & 1 & $Q_{1}^{F},Q_{2}^{1,2|\emptyset}$ & $Q_{1}^{c_1},Q_{2}^{c_1}$& $Q_{1}^{1,2|\emptyset},Q_{2}^{F}$  & $F_5$ & 2 & $Q_{1}^{F},Q_{2}^{2|1}$ & $Q_{1}^{F},Q_{2}^{2|1}$ & $Q_{1}^{2|1},Q_{2}^{F}$ \\[1ex] 
\hline
$F_2$ & 1 & $Q_{1}^{F},Q_{2}^{2|1}$ & $Q_{1}^{F},Q_{2}^{2|1}$ & $Q_{1}^{2|1},Q_{2}^{F}$& $F_6$ & 3 & $Q_{1}^{F},Q_{2}^{1|2}$ & $Q_{1}^{1|2},Q_{2}^{F}$ & $Q_{1}^{1|2},Q_{2}^{F}$\\
\hline
$F_3$ & 1 & $Q_{1}^{F},Q_{2}^{1,2|\emptyset}$ & $Q_{1}^{c_1},Q_{2}^{c_1}$ & $Q_{1}^{1,2|\emptyset},Q_{2}^{F}$& $F_7$ & 8 &  $Q_{1}^{F},Q_{2}^{1|2}$ & $Q_{1}^{1|2},Q_{2}^{F}$ & $Q_{1}^{1|2},Q_{2}^{F}$  \\
\hline
$F_4$ & 1 & $Q_1^F,Q_{2}^{1|2}$ & $Q_{1}^{1|2},Q_{2}^{F}$ & $Q_{1}^{1|2},Q_{2}^{F}$ & $F_8$ & 12 & $Q_{1}^{F},Q_{2}^{2|1}$ & $Q_{1}^{F},Q_{2}^{2|1}$ & $Q_{1}^{2|1},Q_{2}^{F}$ \\[1ex]
\hline
\end{tabular}
\label{table:flexible queues}
\end{table}
In Table~\ref{table:flexible queues}, we define random variable $P[t] \in \{A,B,C\}$ to explain how the flexible destination queues are determined
at time instant $t$. Specifically, we assign the name of different cornerpoints of the stability region as the possible events of $P[t]$ to indicate that Table~\ref{table:flexible queues} follows the rule of the same name cornerpoint to determine the flexible destination queues. This way, we find different values of $P[t]$ as follows: $P[t]=A:$ We know $\lambda_1>\lambda_2$ at cornerpoint A. 
Thus, we determine the destination queues such that ${\sf Tx}_2$ helps ${\sf Tx}_1$ delivering the intended bits to ${\sf Rx}_1$; $P[t]=C:$ In \cite{vahid2014capacity}, authors show that when $0 \leq \delta \leq \frac{3-\sqrt{5}}{2}$, the number of bits in $Q_i^{\bar{i}|i}$ is greater than or equal to the number of bits in $Q_i^{i|\bar{i}}$. Hence, to avoid congestion, we pick the queues that store fewer bits in $Q_i^{\bar{i}|i}$; $P[t]=B:$ We use the same strategy as ${P[t]=A}$ by changing the users' labels. Below, we illustrate the details of flexible destination queues when bits $a$ and $b$ are in the origin queues and $0 \leq \delta \leq \frac{3-\sqrt{5}}{2}$. We will describe the probability of each outcome later.

\noindent $\diamond$ $\text{F}_1$: This case occurs when $Q_1^{1|\emptyset},Q_2^{2|\emptyset}$ are the origin queues and SN-1 happens. At this time, both receivers obtain $a\oplus b$ from the transmitters. Here, there are three choices for destination queues as: $P[t]=A:$ In this case, if ${\sf Tx}_2$ delivers $b$ to both receivers, ${\sf Rx}_2$ gets the desired bit and ${\sf Rx}_1$ decodes $a$ using $b$ and $a\oplus b$. Therefore, $a$ goes to $Q_1^F$ and $b$ saves in $Q_2^{1,2|\emptyset}$. This way, ${\sf Rx}_1$ uses side information from ${\sf Tx}_2$ to decode its desired bit; $P[t]=C:$ Based on the definition of $Q_i^{i|\emptyset}$ in Section~\ref{Subsection:QI}, $a \rightarrow Q_1^{c_1}$ and $b \rightarrow Q_2^{c_1}$. However, we know: 1) In Table~\ref{table:flexible queues}, SN-1 happens most frequently since $0\leq \delta \leq \frac{3-\sqrt{5}}{2}$; 2) $Q_i^{c_1}$ has the lowest priority in Table~\ref{table:priority} and transmitters only need to deliver one of $a$ or $b$ to both receivers. As a result, in half of the communication time, we store the data bit of $Q_i^{c_1}$ in $Q_i^{1,2|\emptyset}$ and move its data bit to $Q_i^F$ in the other half; $P[t]=B:$ By changing the users' labels in $P[t]=A$, we find $Q_1^{1,2|\emptyset},Q_2^F$ as the destination queues;

\noindent $\diamond$ $\text{F}_2$: In this part, the origin queues is one of $\{Q_1^{1|\emptyset},Q_2^{2|1}\}$, $\{Q_1^{2|1},Q_2^{2|\emptyset}\}$, $\{Q_1^{2|1},Q_2^{2|1}\}$, $\{Q_1^{2|1},Q_2^{1,2|\emptyset}\}$, or $\{Q_1^{1,2|\emptyset},Q_2^{2|1}\}$ and SN-1 occurs. We describe the details for $Q_1^{1|\emptyset},Q_2^{2|1}$ as an example. In this case, both receivers obtain $a\oplus b$ which leads to two choices for destination queues as follows: $P[t]= A~\text{\it or}~C:$ In this choice, if ${\sf Tx}_2$ delivers $b$ to its intended receiver, ${\sf Rx}_2$ obtains its desired bit. Then, ${\sf Rx}_1$ uses $a\oplus b$ and side information in $Q_2^{2|1}$ to decode $a$. As the result, $Q_1^{F},Q_2^{2|1}$ are the destination queues. We note that $Q_1^{F},Q_2^{2|1}$ are considered as destination queues for $P[t]= C$ since they do not increase the bits in $Q_i^{\bar{i}|i}$; $P[t]=B:$ In this case, $\lambda_2>\lambda_1$ and if ${\sf Tx}_1$ delivers $a$ to unintended receiver, then ${\sf Rx}_2$ applies $a$ to $a\oplus b$ and attains the desired bit. Then, ${\sf Rx}_1$ uses the side information in $Q_2^{2|1}$ to decode $a$, and $Q_1^{2|1},Q_2^{F}$ are destination queues;

\noindent $\diamond$ $\text{F}_3$: The destination queues of this case are similar to $\text{F}_1$ where $a$ and $b$ are available in one of $\{Q_1^{1|\emptyset},Q_2^{1,2|\emptyset}\}$, $\{Q_1^{1,2|\emptyset},Q_2^{2|\emptyset}\}$ or $\{Q_1^{1,2|\emptyset},Q_2^{1,2|\emptyset}\}$ and SN-1 shows the channel realizations;

\noindent $\diamond$ $\text{F}_4$: This case demonstrates the flexible destination queues when desired bits are available in one of origin queues $\{Q_1^{1|\emptyset},Q_2^{1|2}\}$, $\{Q_1^{1|2},Q_2^{2|\emptyset}\}$, $\{Q_1^{1|2},Q_2^{1|2}\}$, $\{Q_1^{1|2},Q_2^{1,2|\emptyset}\}$ or $\{Q_1^{1,2|\emptyset},Q_2^{1|2}\}$ and SN-1 happens. The procedure of this case is similar to $\text{F}_2$ by changing the users' labels;

\noindent $\diamond$ $\text{F}_5$: Assume that $\{Q_1^{1|\emptyset},Q_2^{2|1}\}$ or $\{Q_1^{1,2|\emptyset},Q_2^{2|1}\}$ are the origin queues and SN-2 occurs. In this case, ${\sf Rx}_1$ receives $a$ and ${\sf Rx}_2$ gets $a\oplus b$. Here, we illustrate the details for $Q_1^{1,2|\emptyset},Q_2^{2|1}$ as an example. We find two choices for destination queues as follows: $P[t]= B:$ In this choice, if ${\sf Tx}_1$ delivers $a$ to its unintended receiver, ${\sf Rx}_2$ gets $a$ easily. Then, ${\sf Rx}_2$ uses $a$ and $a\oplus b$ to decodes $b$. Also, ${\sf Rx}_1$ obtains $a$ directly from ${\sf Tx}_1$. As the result, $a \rightarrow Q_1^{2|1}$ and $b\rightarrow Q_2^{F}$; $P[t]= A~\text{\it or}~C:$ By delivering $b$ to ${\sf Rx}_2$, it applies $b$ to $a\oplus b$ and gets $a$ and ${\sf Rx}_1$ receives $a$ directly from ${\sf Tx}_1$. Therefore, $Q_1^{F}, Q_2^{2|1}$ are the destination queues which do not increase the number of bits in $Q_i^{\bar{i}|i}$. As a result, $Q_1^{F}, Q_2^{2|1}$ are the destination queues for both $P[t]= A$ and $P[t]= C$;

\noindent $\diamond$ $\text{F}_6$: Suppose data bits are available in $\{Q_1^{1|2},Q_2^{2|\emptyset}\}$ or $\{Q_1^{1|2},Q_2^{1,2|\emptyset}\}$ and SN-3 happens. The procedure of this case is similar to $\text{F}_5$ by changing the users' labels;

\noindent $\diamond$ $\text{F}_7$: If $\{Q_1^{1|\emptyset},Q_2^{1|2}\}$ or $\{Q_1^{1,2|\emptyset},Q_2^{1|2}\}$ are the origin queues and SN-8 occurs. We describe the details when $Q_1^{1|\emptyset},Q_2^{1|2}$ are the origin queues as an example. In this case, ${\sf Rx}_1$ gets $a\oplus b$ and ${\sf Rx}_2$ receives $a$. Here are two choices for destination queues: $P[t]= A:$ By delivering $b$ to ${\sf Rx}_1$, it uses $a\oplus b$ and $b$ to decode $a$. Hence, $a\rightarrow Q_1^{F}$ and $b\rightarrow Q_2^{1|2}$; $P[t]= B~\text{\it or}~C:$ If ${\sf Tx}_1$ delivers $a$ to ${\sf Rx}_1$, it gets $a$ and decodes $b$ using $a$ and $a\oplus b$. Hence, $a \rightarrow Q_1^{1|2}$ and $b\rightarrow Q_2^{F}$ which do not increase the bits in $Q_i^{\bar{i}|i}$;

\noindent $\diamond$ $\text{F}_8$: This case contains $\{Q_1^{2|1},Q_2^{2|\emptyset}\}$ or $\{Q_1^{2|1}, Q_2^{1,2|\emptyset}\}$ as the origin queues when SN-12 occurs, and the procedure of this case is similar to $\text{F}_7$ by changing the users' labels.
\subsection{Finding the corresponding cornerpoint in Table~\ref{table:flexible queues}}
\label{subsect:find_stability_zones}
In this part, we explain how our proposed transmission protocol generates $P[t]$. We define $p_A$, $p_B$ and $p_C$ as the probability of $A$, $B$ and $C$, respectively. 
It is clear that $0\leq p_A,p_B,p_C\leq 1$ and $p_A+p_B+p_C=1$. We use ${\lambda}_i^A$, ${\lambda}_i^B$, and ${\lambda}_i^C$ to denote the rate at cornerpoint A, B, and C at user $i$, respectively. To describe the above probabilities, we divide stable throughput region into five different zones (Fig.~\ref{Fig:zones}) as follows:
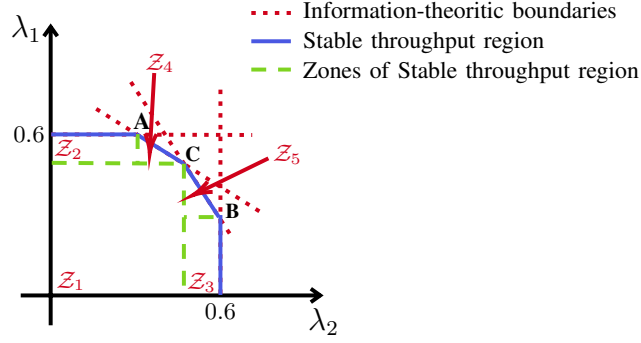
\begin{figure}[!ht]
  \centering

\tikzset{every picture/.style={line width=0.6pt}} 

\begin{tikzpicture}[x=0.6pt,y=0.6pt,yscale=-0.85,xscale=0.85]

\draw [line width=1.5]  (67.5,260.66) -- (289.64,260.66)(89.71,61.25) -- (89.71,282.82) (282.64,255.66) -- (289.64,260.66) -- (282.64,265.66) (84.71,68.25) -- (89.71,61.25) -- (94.71,68.25)  ;
\draw [color={rgb, 255:red, 208; green, 2; blue, 27 }  ,draw opacity=1 ][line width=1.5]  [dash pattern={on 1.69pt off 2.76pt}]  (215.62,107.9) -- (215.62,260.29) ;
\draw [color={rgb, 255:red, 208; green, 2; blue, 27 }  ,draw opacity=1 ][line width=1.5]  [dash pattern={on 1.69pt off 2.76pt}]  (90.14,141.13) -- (239.19,141.57) ;
\draw [color={rgb, 255:red, 208; green, 2; blue, 27 }  ,draw opacity=1 ][line width=1.5]  [dash pattern={on 1.69pt off 2.76pt}]  (149.19,102.24) -- (223.81,216.8) ;
\draw [color={rgb, 255:red, 208; green, 2; blue, 27 }  ,draw opacity=1 ][line width=1.5]  [dash pattern={on 1.69pt off 2.76pt}]  (123.5,122.18) -- (245.5,198.85) ;
\draw [color={rgb, 255:red, 74; green, 88; blue, 226 }  ,draw opacity=1 ][line width=1.5]    (215.62,201) -- (215.62,260.29) ;
\draw [color={rgb, 255:red, 74; green, 88; blue, 226 }  ,draw opacity=1 ][line width=1.5]    (153.19,140.9) -- (188.52,162.9) ;
\draw [color={rgb, 255:red, 74; green, 88; blue, 226 }  ,draw opacity=1 ][line width=1.5]    (188.52,162.9) -- (214.5,202.5) ;
\draw [color={rgb, 255:red, 126; green, 211; blue, 33 }  ,draw opacity=1 ][line width=1.5]  [dash pattern={on 5.63pt off 4.5pt}]  (90.52,162.24) -- (188.52,162.9) ;
\draw [color={rgb, 255:red, 126; green, 211; blue, 33 }  ,draw opacity=1 ][line width=1.5]  [dash pattern={on 5.63pt off 4.5pt}]  (188.52,162.9) -- (188.52,261.57) ;
\draw [color={rgb, 255:red, 126; green, 211; blue, 33 }  ,draw opacity=1 ][line width=1.5]  [dash pattern={on 5.63pt off 4.5pt}]  (154.19,140.9) -- (154.19,162.9) ;
\draw [color={rgb, 255:red, 126; green, 211; blue, 33 }  ,draw opacity=1 ][line width=1.5]  [dash pattern={on 5.63pt off 4.5pt}]  (210.52,202.5) -- (188.52,202.5) ;
\draw [color={rgb, 255:red, 74; green, 88; blue, 226 }  ,draw opacity=1 ][fill={rgb, 255:red, 74; green, 144; blue, 226 }  ,fill opacity=1 ][line width=1.5]    (90.14,141.13) -- (153.05,141.13) ;
\draw [color={rgb, 255:red, 208; green, 2; blue, 27 }  ,draw opacity=1 ][line width=1.5]    (251.19,158.9) -- (197.89,184.92) ;
\draw [shift={(195.19,186.24)}, rotate = 333.98] [color={rgb, 255:red, 208; green, 2; blue, 27 }  ,draw opacity=1 ][line width=1.5]    (14.21,-4.28) .. controls (9.04,-1.82) and (4.3,-0.39) .. (0,0) .. controls (4.3,0.39) and (9.04,1.82) .. (14.21,4.28)   ;
\draw [color={rgb, 255:red, 208; green, 2; blue, 27 }  ,draw opacity=1 ][fill={rgb, 255:red, 208; green, 2; blue, 27 }  ,fill opacity=1 ][line width=1.5]    (166.52,95.57) -- (163.36,151.91) ;
\draw [shift={(163.19,154.9)}, rotate = 273.22] [color={rgb, 255:red, 208; green, 2; blue, 27 }  ,draw opacity=1 ][line width=1.5]    (14.21,-4.28) .. controls (9.04,-1.82) and (4.3,-0.39) .. (0,0) .. controls (4.3,0.39) and (9.04,1.82) .. (14.21,4.28)   ;
\draw [color={rgb, 255:red, 208; green, 2; blue, 27 }  ,draw opacity=1 ][line width=1.5]  [dash pattern={on 1.69pt off 2.76pt}]  (238.52,51) -- (267.86,51) ;
\draw [color={rgb, 255:red, 74; green, 88; blue, 226 }  ,draw opacity=1 ][fill={rgb, 255:red, 74; green, 144; blue, 226 }  ,fill opacity=1 ][line width=1.5]    (238.52,72) -- (267.86,72) ;
\draw [color={rgb, 255:red, 126; green, 211; blue, 33 }  ,draw opacity=1 ][line width=1.5]  [dash pattern={on 5.63pt off 4.5pt}]  (238.52,93) -- (267.86,93) ;

\draw (60,133) node [anchor=north west][inner sep=0.75pt]  [font=\footnotesize]  {${\textstyle 0.6}$};
\draw (202,265) node [anchor=north west][inner sep=0.75pt]  [font=\footnotesize]  {${\textstyle 0.6}$};
\draw (149.2,123.17) node [anchor=north west][inner sep=0.75pt]  [font=\footnotesize] [align=left] {\textbf{A}};
\draw (187.33,148.24) node [anchor=north west][inner sep=0.75pt]  [font=\scriptsize] [align=left] {\textbf{C}};
\draw (217.33,190.24) node [anchor=north west][inner sep=0.75pt]  [font=\scriptsize] [align=left] {\textbf{B}};
\draw (92,240.9) node [anchor=north west][inner sep=0.75pt]  [font=\footnotesize]  {$\textcolor[rgb]{0.82,0.01,0.11}{\mathcal{Z}_{1}}$};
\draw (90.14,143) node [anchor=north west][inner sep=0.75pt]  [font=\footnotesize]  {$\textcolor[rgb]{0.82,0.01,0.11}{\mathcal{Z}_{2}}$};
\draw (189,241.67) node [anchor=north west][inner sep=0.75pt]  [font=\footnotesize]  {$\textcolor[rgb]{0.82,0.01,0.11}{\mathcal{Z}_{3}}$};
\draw (158,79) node [anchor=north west][inner sep=0.75pt]  [font=\footnotesize]  {$\textcolor[rgb]{0.82,0.01,0.11}{\mathcal{Z}_{4}}$};
\draw (252.67,147) node [anchor=north west][inner sep=0.75pt]  [font=\footnotesize]  {$\textcolor[rgb]{0.82,0.01,0.11}{\mathcal{Z}_{5}}$};
\draw (58,50.57) node [anchor=north west][inner sep=0.75pt]    {${\lambda}_1$};
\draw (278.67,268) node [anchor=north west][inner sep=0.75pt]    {${\lambda}_2$};
\draw (275,41) node [anchor=north west][inner sep=0.75pt]   [align=left] {{\footnotesize Information-theoritic boundaries}};
\draw (275,63) node [anchor=north west][inner sep=0.75pt]   [align=left] {{\footnotesize Stable throughput region}};
\draw (275,85) node [anchor=north west][inner sep=0.75pt]  [font=\tiny] [align=left] {{\footnotesize Zones of Stable throughput region}};
\end{tikzpicture}  
  \caption{Different zones of stable throughput region for homogeneous setting with $\delta=0.4$}.
  \label{Fig:zones}
  \vspace{-12mm}
\end{figure}

$\mathcal{Z}_1:$ This region in the convex hull of the following points:  $(0,0)$, $({\lambda}_1^C,0)$, $(0,{\lambda}_2^C)$ and $({\lambda}_1^C,{\lambda}_2^C)$. We note that under certain conditions, the cornerpoint C could be equal to $(1-\delta_{11}, 1-\delta_{22})$ in which case the entire region is covered by $\mathcal{Z}_1$;
    
$\mathcal{Z}_2:$ This zone is a convex hull which is confined by $({\lambda}_1^C,0)$, $({\lambda}_1^A,0)$,
$({\lambda}_1^A,{\lambda}_2^A)$ and $({\lambda}_1^C,{\lambda}_2^A)$. We use $p_A=1$ in this region since ${\lambda}_1>{\lambda}_2$ and arrival rates are not close to the boundary points between cornerpoint A and C;

$\mathcal{Z}_3:$ This zone contains a convex hull which is surrounded by $(0,{\lambda}_2^C)$, $(0,{\lambda}_2^B)$,
$({\lambda}_1^B,{\lambda}_2^B)$ and $({\lambda}_1^B,{\lambda}_2^C)$. Similar to $\mathcal{Z}_2$, we determine the corresponding cornerpoint using $p_B=1$;

$\mathcal{Z}_4:$ This area is in the convex hull which is confined by cornerpoints A and C and $({\lambda}_1^C,{\lambda}_2^A)$. We note that $\mathcal{Z}_4$ includes the boundary points between cornerpoints A and C which lead to $p_C\neq 0$ and $p_A \neq 0$. We use linear equations in (\ref{eq:P_C}) and (\ref{eq:P_AP_B}) to describe the relation between $p_A$ and $p_C$ such that $p_A=1$ and $p_C=1$ at cornerpoints A and C, respectively;
    
$\mathcal{Z}_5:$ This zone shows the area between $\mathcal{Z}_1$, $\mathcal{Z}_3$ and the outer-bound of stable throughput region which contains the boundary points between cornerpoints C and B. Similar to $\mathcal{Z}_4$, in this area, $p_C\neq 0$ and $p_B \neq 0$ and they are linearly related based on (\ref{eq:P_C}) and (\ref{eq:P_AP_B}).

According to these different zones, we define $p_C$ as
\begin{align}
\label{eq:P_C}
p_C=
\left\{ \begin{array}{lll}
\frac{\Delta_Z-\Delta {\lambda}}{\Delta_Z-| {\lambda}_2^C-\lambda_1^C}|, & ( {\lambda}_1,\lambda_2) \in \mathcal{Z}_4~\text{or}~\mathcal{Z}_5, \\
1, & ({\lambda}_1,\lambda_2) \in \mathcal{Z}_1, 
\\
0, & ({\lambda}_1,\lambda_2) \in \mathcal{Z}_2~\text{or}~\mathcal{Z}_3,
\end{array} \right.
\text{,~where}~& \Delta_Z=
\left\{ \begin{array}{ll}
|{\lambda}_2^A-\lambda_1^A|, & ({\lambda}_1,\lambda_2) \in \mathcal{Z}_4, \\
|{\lambda}_2^B-\lambda_1^B|, & ({\lambda}_1,\lambda_2) \in \mathcal{Z}_5,
\end{array} \right.
\end{align}
and where $\Delta{\lambda}=|\lambda_2-\lambda_1|$, and $p_A$ and $p_B$ are written as
\begin{align}
\label{eq:P_AP_B}
p_A=
\left\{ \begin{array}{lll}
1-p_C, & ({\lambda}_1,\lambda_2) \in \mathcal{Z}_4, \\
1, & ({\lambda}_1,\lambda_2) \in \mathcal{Z}_2,
\\
0, & \text{otherwise},
\end{array} \right.
& p_B=
\left\{ \begin{array}{lll}
1-p_C, & ({\lambda}_1,\lambda_2) \in \mathcal{Z}_5, \\
1, & ({\lambda}_1,\lambda_2) \in \mathcal{Z}_3,
\\
0, & \text{otherwise}.
\end{array} \right.
\end{align}

According to (\ref{eq:P_C}) and (\ref{eq:P_AP_B}), there are at most two non-zero probabilities in each zone that complement each other. Therefore, our transmission protocol can generate $P[t]$ with Bernoulli distribution. However, to track the destination queues, all network nodes need to get the same outcome from $P[t]$ at each time instant. To do this, network nodes will generate $P[t]$ in a ``deterministic'' fashion using a pseudo-random number generator with a common starting seed. 
\begin{remark}
\label{Remark:SlowRateChange}
We note each node determines the current zone of Fig.~\ref{Fig:zones} based on $(\lambda_1, \lambda_2)$. We assume the arrival rates are stationary and may change slowly over time. In particular, the arrival rates remain constant for a sufficiently large number of time instants such that the system reaches stability. Then, if ${\sf Tx}_i$ finds its arrival rate has changed, it informs other nodes of this change during the update procedure as described in Section~\ref{subsect:update}. To ensure all bits are treated by the delivery time according to the same transmission policy, when rates change, transmitters can halt communications for $o(n)$ as further described when we study bit lifetime in Section~\ref{sub:life_time}. We note that this pause will not affect the overall rates when $n$ is sufficiently large.
\end{remark}
Moreover, to calculate $P[t]$ for any $(\lambda_1,\lambda_2)\in \mathcal{D}$, all nodes: 1) Find the current stability zone based on $(\lambda_1,\lambda_2)$; 2) Use \eqref{eq:P_C} and \eqref{eq:P_AP_B} to obtain $p_A$, $p_B$ and $p_C$; 3) Run a pseudo-random number generator with a common starting seed to generate $P[t]$.

\subsection{Tracking the status of the transmitted bits}
\label{subsect:Transparent_CT}

In this part, we show that if all signals are delivered successfully, ${\sf Rx}_{1}$ will have sufficient number of linearly independent equations to decode all data bits. Moreover, we prove that each network node can track the status of any transmitted bit with no confusion. To do this, we present the following remark to explain different types of delivered data bits:
\begin{remark}
\label{Remark:ActualDelivery}
Data bits are delivered through: 1) Actual delivery: ${\sf Rx}_i$ obtains the bit directly from ${\sf Tx}_i$ with no interference; 2) Virtual delivery: ${\sf Rx}_i$ needs future bits from ${\sf Tx}_{\bar{i}}$ to recover its own bit; 3) XORed delivery: In this case, the non-XORed bit gets XOR with other bit, and then the XORed bit may get another XOR combination or be delivered actually or virtually. We say that virtual and XORed delivery describe the necessary equations for decodable bits at ${\sf Rx}_i$.
\end{remark}  

\begin{claim}
Assuming successful delivery of all signals according to the strategy described above, ${\sf Rx}_i,i=1,2$ has sufficient number of linearly independent equations to decode all its desired data bits.
\end{claim}
\begin{proof}
We prove the claim for ${\sf Rx}_{1}$, and a similar statement holds for ${\sf Rx}_{2}$. It suffices to prove the claim for the XORed bits as actual/virtual delivery (see Remark~\ref{Remark:ActualDelivery}) do not require further action.

First, we show that the claim holds at $t=1$. Data bit is available in $Q_1^{1|\emptyset}$ and can be delivered through actual or virtual delivery. In Section~\ref{Subsection:DO}, we showed how ${\sf Tx}_{2}$ helps ${\sf Rx}_{1}$ to recover its desired bit in virtual delivery, and thus, ${\sf Rx}_{1}$ gets sufficient number of linearly independent equations to recover the desired bit;

Second, we show that if the claim holds for time instant $t$, then it also holds for time instant $t+1$. Suppose, ${\sf Rx}_i,i=1,2$ has sufficient number of linearly independent equations to decode all data bits at the end of time instant $t$. At time instant $t+ 1$, ${\sf Tx}_{1}$ can deliver the data bit through actual, virtual, or XORed delivery to ${\sf Rx}_{1}$. As we discussed in Sections~\ref{Subsection:DO} and \ref{Subsection:XOR_COM}, by delivering participated bits in delivering options and XOR models, ${\sf Rx}_{1}$ gets sufficient number of linearly independent equations to recover its desired bits. This completes the proof.
\end{proof}

In addition, to show all bits are tractable at each network node, we assume the control and the priority tables are shared with all nodes prior to communication. As detailed in Section~\ref{subsect:update}, the transmitters regularly announce the number of newly arrived bits. Then, we have the following claim whose proof is presented in Appendix~\ref{apndx:transparent_CT}. 

\begin{lemma}
\label{claim:transparent_CT}
Given the available information at each node as described above, the status of all bits at each transmitter can be tracked correctly at any other node.
\end{lemma}
\subsection{Pseudo code of protocol}
\label{subsect:pseudo-code}

We present a pseudo code, Algorithm~\ref{algo:general_pseudo}, to explain how ${\sf Tx}_i$ implements the general transmission protocol. At each time instant, newly arrived bits go to $Q_i^{i|\emptyset}$, and ${\sf Tx}_i$ selects the origin queues with the highest priority. Later with delayed CSI, ${\sf Tx}_i$ determines the destination queues. Finally, ${\sf Tx}_i$ checks the XOR opportunities and if available, it stores the XORed bit in $Q_i^{1,2|\emptyset}$ and removes the corresponding bits from their queues. 
\begin{algorithm}
  \caption{Pseudo code for general transmission protocol at time instant $t$}
  \begin{algorithmic}[1]
  \label{algo:general_pseudo}
  \STATE Save newly arrived bits in $Q_{i}^{i|\emptyset}$;\
  \STATE Receive $S[t-1]$ as delayed CSIT from both receivers;\
  \STATE Determine the destination queues using $S[t-1]$ and the origin queues of time instant $t-1$ such that
  \IF{ the origin queues belong to $\{\text{F}_1,\text{F}_2, \ldots,\text{F}_8\}$}
  \STATE Use Table~\ref{table:flexible queues} to find the destination queues;
  \ELSIF{there is more than one choice for destination queues}
  \STATE {Choose destination queues that provide more XOR opportunities based on our general rule in Section~\ref{Subsection:CT}};
  \ELSE
  \STATE {Select the only option for destination queues;}
  \ENDIF
  \IF{XOR combination happens}
    \STATE Create XORed bit and save it in $Q_{i}^{1,2|\emptyset}$;\
    \STATE Remove the bits (which participate in XOR combination) from their queues.\\
  \ENDIF
  \STATE Select the origin queues at time instant $t$ with highest priority based on Table~\ref{table:priority};\
  \STATE Transmit the data bits from the origin queues.
\end{algorithmic}
\end{algorithm}
\section{Transmission protocol for non-homogeneous settings}
\label{sect:non-homogeneous}
So far we explained the details of the transmission protocol for the homogeneous settings. In this section, we focus on non-homogeneous settings.
\subsection{Symmetric Settings}
In this case, we assume that the erasure probabilities of direct channels from ${\sf Tx}_i$ to ${\sf Rx}_i$ are equal $(i.e.~ \delta_{ii}=\delta_d,~ i=1,2)$, and similarly for cross channels from ${\sf Tx}_i$ to ${\sf Rx}_{\bar{i}}$, we have $\delta_{i\bar{i}}=\delta_c,~i=1,2,~ \bar{i}=3-i$. In this work, we focus on $0\leq\delta_d\leq 0.5$ and $0\leq\delta_c\leq \frac{1}{2-\delta_d}$ and observe that, compared to transmission protocol in homogeneous settings, we only need to change the flexible destination queues. Here, we explain the details when $0\leq\delta_d\leq \frac{3-\sqrt{5}}{2}$ as follows:

\subsubsection{$0\leq \delta_c < \delta_d$}
In this case, all channels have low erasure probabilities and cross channels are stronger than direct channels. Therefore, we update all flexible destination queues such that more bits go to $Q_i^{\bar{i}|i}$ when $P[t]= C$. 
To clarify the procedure, we consider an example. Suppose $Q_1^{1,2|\emptyset}, Q_2^{2|1}$ are the origin queues and $\text{F}_5$ happens. As we discussed in Table~\ref{table:flexible queues}, $Q_1^{2|1}, Q_2^{F}$ and $Q_1^{F}, Q_2^{2|1}$ are two choices for destination queues. Here, we select $Q_1^{2|1}, Q_2^{F}$ which lead more bits to $Q_i^{\bar{i}|i}$. Similar statement also holds for other flexible cases;
\subsubsection{$\delta_d < \delta_c \leq \frac{1}{2-\delta_d}$}
In this interval, direct channels are stronger than cross links. Thus, we deliver more bits through the direct channels, and this leads to similar strategy as Table~\ref{table:flexible queues}.

Similarly, we can obtain the destination queues when $\frac{3-\sqrt{5}}{2}\leq\delta_d \leq 0.5$. 
\subsection{Non-homogeneous settings}
In this part, we consider all channels have different erasure probabilities and design the transmission protocol based on different $\delta_{ij}$ values. To do this, we extend the flexible destination queues to all cases which deal with more than one option for destination queues. In this work, we focus on the case when  $\frac{3-\sqrt{5}}{2}<\delta_{ii}\leq 0.5$, $0 \leq \delta_{\bar{i}\bar{i}}\leq \frac{3-\sqrt{5}}{2}$, $0.5 \leq \delta_{i\bar{i}}< \frac{1}{2-\delta_{ii}}$ and $\delta_{ii} < \delta_{\bar{i}i} \leq 0.5$. 
We note that flexible cases when $P[t]\neq C$ are similar to Table~\ref{table:flexible queues}. However, we observe that non-symmetric channel gains lead to non-symmetric cornerpoint C. Thus, we select the queues that help data bits to be delivered through the stronger channels in future. In other words, we choose the destination queues correspond to the links with lower erasure probabilities. Suppose an example to clarify this procedure where $a$ and $b$ are available in $Q_1^{1,2|\emptyset}, Q_2^{2|1}$, above channel conditions hold for $i=1$, and SN-2 happens. We can show that there are two options for destination queues: 1) $Q_1^{F}, Q_2^{2|1}$: In  this option, ${\sf Tx}_2$ delivers $b$ to ${\sf Rx}_2$ through $S_{22}[t]$; 2) $Q_1^{2|1}, Q_2^{F}$: In this case, ${\sf Tx}_1$ delivers its data to unintended receiver through the cross link. Based on channel gains, we select $Q_1^{F}, Q_2^{2|1}$ since $S_{22}[t]$ is stronger than $S_{12}[t]$. Similar statements also hold for other flexible destination queues.
\begin{remark}
We note that Algorithm~\ref{algo:general_pseudo} describes our general transmission protocol for any erasure probabilities and the non-homogeneous setting. To do this, we only need to update the flexible destination queues based on different $\delta_{ij}$ values.
\end{remark}
\section{Simulation}
\label{Section:Simulation}
In this section, we provide simulation results of our proposed scheme (the corresponding  Python code is available online at~\cite{key}.
\subsection{Stable Throughput Region vs. Capacity Region}

We compare the stable throughput region of our proposed scheme for stochastic data bit arrivals under non-homogeneous channel assumption with our generalized information-theoretic outer-bounds with delayed-CSIT in~\eqref{eq:outer-bound}. 
We make the following assumptions in our simulations: 1) Data bits arrive at ${\sf Tx}_i$ according to a Poisson $(\lambda_i \in [0, 1])$ distribution, $i=1,2$, where $\lambda_i$ increases in $0.001$ increments and each transmitter can send at most one bit per time instant; 2) Channel coefficient of each link between ${\sf Tx}_i$ and ${\sf Rx}_j$ at each time instant is a Bernoulli random variable ($S_{ij}(t)\sim \mathcal{B}(1-\delta_{ij})$) and is distributed independently from other users and across time; 3) We consider two different regimes based on the results of~\cite{vahid2014capacity} and our generalized outer-bounds  in~\eqref{eq:outer-bound}, to select the values of $\delta_{ij}$ in non-homogeneous channels as: i) $0\leq \delta_{i\bar{i}}\leq \frac{1}{2-\delta_{ii}}$, $i=1,2, \bar{i}=3-i$ where capacity region is known; ii) $\frac{1}{2-\delta_{ii}}\leq \delta_{i\bar{i}}\leq 1$ as the unknown capacity region; 4) We determine the destination queues with respect to the values of $\lambda_i$ and $\delta_{ij}$, and bits with highest priority are sent at each time instant; 5) We use $\beta$ to denote the average time instant that a bit stays in the communication network and show by simulation that $\beta \propto \sqrt{n}$. We assume data bits are arriving during $n-\beta$ time instants to the transmitters and dedicate last $\beta$ time instants of communication time to deliver the remained bits in the network queues; 6) We set acceptable error $\epsilon$ equal to $1\%$ and say each rate-tuple ($\hat{\lambda}_1, \hat{\lambda}_2$) that satisfies definition~\ref{def:stable_rate} is considered to be stable. In these cases, we verify through simulations that by increasing the communication block length, the error margin in fact decreases.




Fig.~\ref{Fig:ordinary_pross_rate_p_0.6} depicts the stable throughput region of the non-homogeneous setting as well as the information-theoretic outer-bounds when $\delta_{11}=0.4$, $\delta_{12}=0.6$, $\delta_{22}=0.2$, $\delta_{21}=0.5$, $n=3\times 10^4$, $\beta=185$, and $\epsilon=1\%$ . This figure shows that all cornerpoints are achievable and the stable throughput region \underline{matches} the capacity region under our proposed scheme. Also, Fig.~\ref{Fig:ordinary_pross_rate_p_0.6} shows that if similar data movement rules of cornerpoint B are used at cornerpoint A, then the stable throughput region is strictly smaller than the capacity region. As a result, data movement rules in non-homogeneous setting should be changed with respect to the relative values of $\lambda_i$'s and $\delta_{ij}$'s.
\begin{figure}[!ht]
  \centering
  \includegraphics[trim = 0mm 0mm 0mm 0mm, clip, scale=5.5, width=0.55\linewidth]{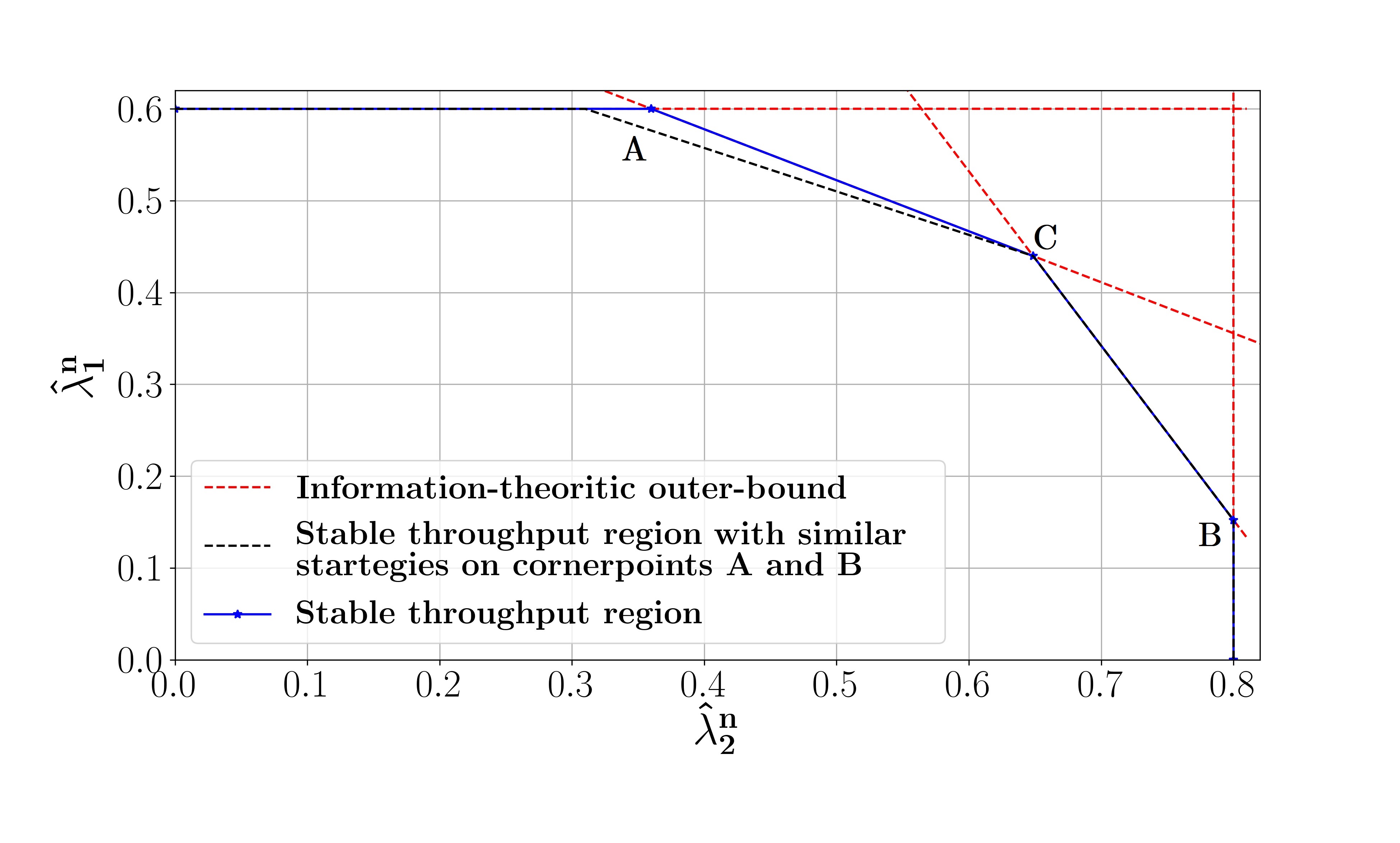}
  \vspace{-5mm}
  \caption{\it Comparing the stable throughput region and the capacity region in the non-homogeneous setting when $\delta_{11}=0.4$, $\delta_{12}=0.6$, $\delta_{22}=0.2$, $\delta_{21}=0.5$, $n=3\times 10^4$, $\beta=185$ and $\epsilon=1\%$.} \label{Fig:ordinary_pross_rate_p_0.6}
  \vspace{-2mm}
\end{figure}

Fig.~\ref{Fig:epsilon}(a) describes the stable throughput region for a case in which the capacity region is unknown. More specifically, we set $\delta_{11}=0.4$, $\delta_{12}=0.7$, $\delta_{22}=0.4$, $\delta_{21}=0.7$, $n=3\times 10^4$, $\beta=185$, and $\epsilon=1\%$. This figure shows that stable throughput region \underline{does not} match with information theoretic outer-bounds and there is a $1.5\%$ error margin at the cornerpoint C that does not seem to diminish as we increase the block length. 

Fig.~\ref{Fig:epsilon}(b) shows the impact of the value of the acceptable error rate, $\epsilon$, at the receivers on the shape of the stable throughput region. In particular, we consider two different values of $\epsilon$ ($1\%$ and $0.1\%$) in the homogeneous setting when $\delta=0.4$, $n=2\times 10^4$ and $\beta=140$. As we observe, the overall shape of the stable throughput region remains similar but gets closer to the outer-bounds as $\epsilon$ increases.
\begin{figure}[!ht]
  \centering
  \includegraphics[trim = 0mm 0mm 0mm 0mm, clip, scale=7.5, width=0.7\linewidth]{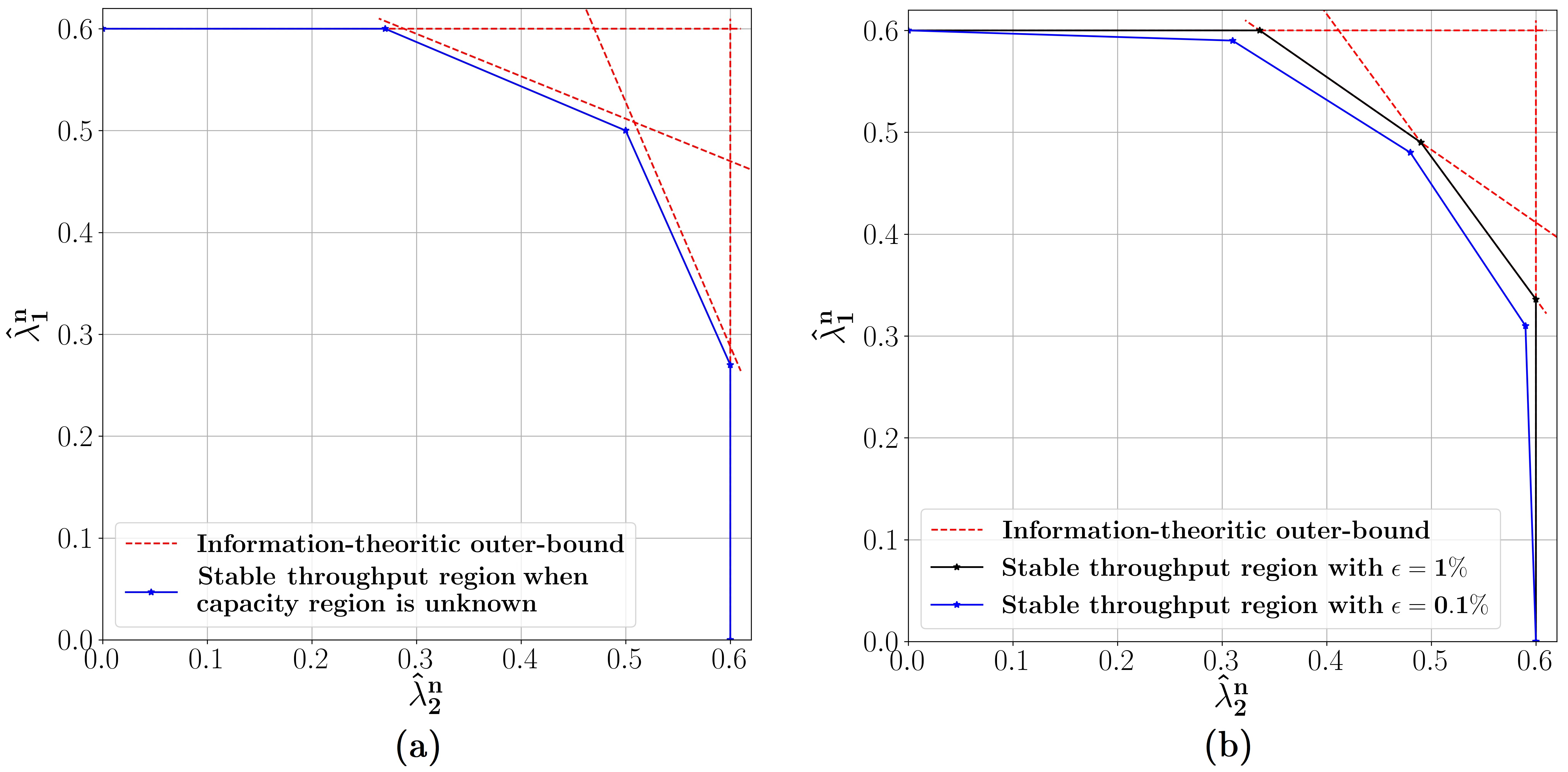}
  \caption{\it (a) The stable throughput region and the information-theoretic outer-bounds when the capacity region is unknown: $\delta_{11}=0.4$, $\delta_{12}=0.7$, $\delta_{22}=0.4$, $\delta_{21}=0.7$, $n=3\times 10^4$, $\beta=185$, and $\epsilon=1\%$; 
  (b) The stable throughput regions for different values of $\epsilon$ (1$\%$ and 0.1$\%$) versus the information-theoretic outer-bounds when $\delta=0.4$, $n=2\times 10^4$, and $\beta=140$.} \label{Fig:epsilon}
  \vspace{-5mm}
\end{figure}
\subsection{Coding Complexity}

The power of wireless is in multi-casting, and thus, the key idea behind optimal feedback-based transmission strategies for BC and distributed interference channels is to create as many linear combinations as possible in order to simultaneously satisfy multiple users. However, increasing the number of XOR operations raises encoding and decoding complexity. Thus, we analyze the number of XOR operations in our proposed protocol to gain a deeper understanding of this tradeoff. Fig.~\ref{Fig:num_XOR_comb}(a) describes a histogram of the number of bits which are delivered through XOR operations for the cornerpoint C in homogeneous setting when $\delta=0.4$, $n=10^4$ and $\beta=100$. We observe  that about $40$ percent of total bits are delivered through the XOR combinations. If we use an encoder-decoder mechanism which does not allow linear combinations with more than 3 participating bits, about $25$ percent of total bits will be affected based on Fig.~\ref{Fig:num_XOR_comb}(a). It is thus an interesting future direction to characterize the stable throughput region of wireless networks with limitations on encoding and decoding complexity.

\begin{figure}[!ht]
  \centering
  \includegraphics[trim = 0mm 0mm 0mm 0mm, clip, scale=5.5, width=.99\linewidth]{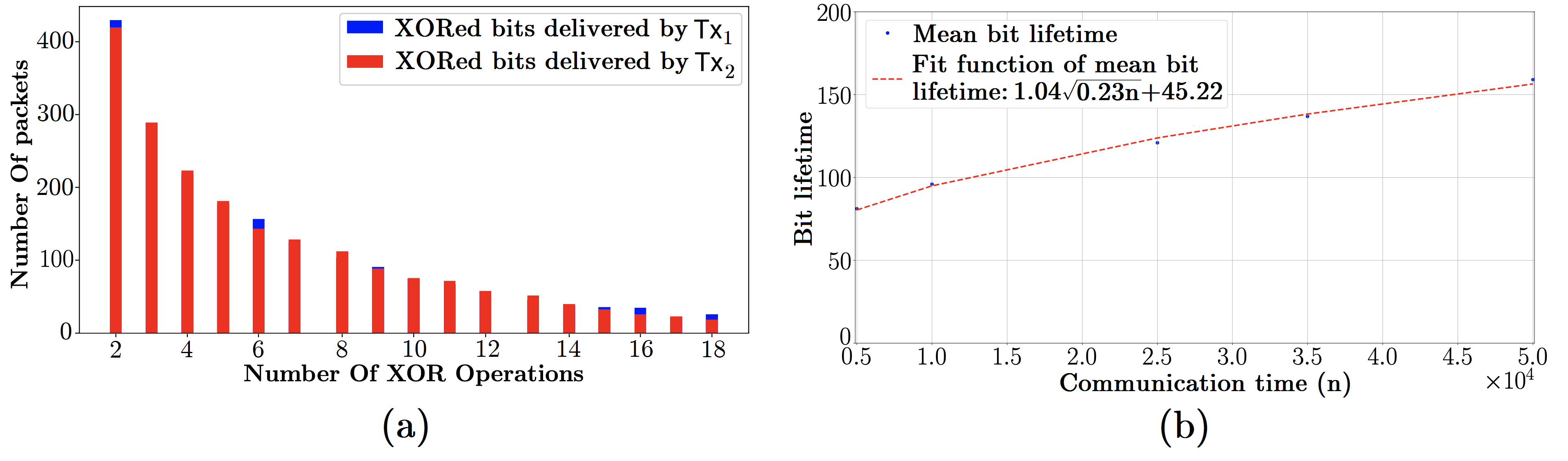}
  \caption{\it (a) Number of bits versus number of XOR operations when $\delta=0.4$, $n=10^4$, $\beta=100$ and $\epsilon=1\%$; (b) Mean value of ${\sf Tx}_1$'s bit lifetime and  $1.04\sqrt{0.23n}+45.22$ as its fitting curves versus communication time for cornerpoint C when $\delta=0.4$ and $n\in[5\times 10^3,5\times 10^4]$.\label{Fig:num_XOR_comb}}
\end{figure}
\subsection{Bit Lifetime (delivery delay)}
\label{sub:life_time}

In general, some data bits may have delivery deadline that transmission protocols must take into account. To evaluate this delay for our protocol, we define the lifetime of a bit as the time interval between its arrival to the system (joins $Q_{i}^{i|\emptyset}$) and its actual delivery time (either be delivered with no interference or sufficient number of linearly independent equations join $Q_{i}^{F}$ so that the bit can be recovered). Fig.~\ref{Fig:num_XOR_comb}(b) depicts the mean value of bit lifetime for the cornerpoint C in the homogeneous setting when $\delta=0.4$ and $n\in[5\times 10^3,5\times 10^4]$. This figure shows the mean value of bit lifetime follows the square root of the total communication time. 
We also learn from our simulations that the worst-case lifetime behaves linearly in the communication time.

\section{Conclusion}
\label{Section:Conclusion}
 We studied the stable throughput region of an interference channel with two transmitter-receiver pairs under the non-homogeneous setting, and 
 showed numerically that the stable throughput region, under specific assumptions, matches the known capacity region. We also examined our proposed scheme in other cases and for some examples numerically showed the gap between the stable throughput region and the outer-bounds. We presented general transmission protocols with respect to $\delta_{ij}$ which achieve the entire region. We defined a priority policy to determine which bit should be sent at any given time in order to maximize the achievable stable rates. Also, we proposed multiple control tables with respect to the asymmetric channel conditions and define flexible cases to manage the data movements between different queues. Finally, we showed the bit lifetime in our protocol scales as the square root of the total communication time, and we investigated how many XOR operations are needed to deliver the data bits. An interesting future direction is to quantify the trade-off we illustrated between complexity and stable rates. One possible approach would be to incorporate techniques such as Reinforcement Learning to solve this problem.
\begin{appendices}
\section{Proof of theorem~\ref{thm:outer-bound_capacity} }
\label{apndx:A1}
\begin{proof}
Assume rate-tuple $(\hat{\lambda}_1,\hat{\lambda}_2)$ is achievable. First, we note that the condition of $\hat{\lambda}_i\leq 1-\delta_{ii}, i=1,2$, is trivial since the transmitters do not exchange their messages and ${\sf Rx}_i$ can receive at most $(1-\delta_{ii})$ of transmitted bits from ${\sf Tx}_i$ when $S_{ii}[t]=1$. Second, we define $\alpha_{i}=\frac{1-\delta_{ii}\delta_{i\bar{i}}}{1-\delta_{i\bar{i}}}, i=1,2, \bar{i}=3-i$ to derive the outer-bound. We have
\begin{align}
\label{eq:start_out_proof}
n(\hat{\lambda}_1+\alpha_1\hat{\lambda}_2)& =H(W_1)+\alpha_1H(W_2)\stackrel{(a)}{=}H(W_1|W_2,S^n)+\alpha_1H(W_2|S^n) \\ \nonumber
&\stackrel{\text{Fano}}{\leq} I(W_1;Y_1^n|W_2,S^n)+\alpha_1I(W_2;Y_2^n|S^n)+n\xi_n \\ \nonumber
& = H(Y_1^n|W_2,S^n)-\underbrace{H(Y_1^n|W_1,W_2,S^n)}_{=~0}+\alpha_1H(Y_2^n|S^n)-\alpha_1H(Y_2^n|W_2,S^n)+n\xi_n \\ \nonumber
&\stackrel{(b)}{=}\alpha_1H(Y_2^n|S^n)+H(Y_1^n|W_2,X_2^n,S^n)-\alpha_1H(Y_2^n|W_2,X_2^n,S^n)+n\xi_n \\ \nonumber
&\stackrel{(c)}{=}\alpha_1H(Y_2^n|S^n)+H(S_{11}^nX_1^n|W_2,X_2^n,S^n)-\alpha_1H(S_{12}^nX_1^n|W_2,X_2^n,S^n)+n\xi_n \\ \nonumber
&\stackrel{(d)}{=}\alpha_1H(Y_2^n|S^n)+H(S_{11}^nX_1^n|W_2,S^n)-\alpha_1H(S_{12}^nX_1^n|W_2,S^n)+n\xi_n \\ \nonumber
&\stackrel{(e)}{=}\alpha_1H(Y_2^n|S^n)+H(S_{11}^nX_1^n|S^n)-\alpha_1H(S_{12}^nX_1^n|S^n)+n\xi_n \\ \nonumber
&\stackrel{(f)}{\leq}\alpha_1H(Y_2^n|S^n)+n\xi_n=\alpha_1\sum_{t=1}^n  H(Y_2[t]|Y_2^{t-1},S^n)+n\xi_n  \\ \nonumber
&\stackrel{(g)}{\leq}\alpha_1\sum_{t=1}^n  H(Y_2[t]|S^n)+n\xi_n \stackrel{(h)}{\leq} 
\frac{(1-\delta_{11}\delta_{12})(1-\delta_{22}\delta_{12})}{1-\delta_{12}}n+n\xi_n
\end{align}
where $(a)$ holds since $S^n$, $W_1$, and $W_2$ are mutually independent; $(b)$ is true since $X_2^n$ is a function of $(W_2,S^n)$; $(c)$  holds since $X_2^n$ and $S^n$ are known; and $(d)$ happens due the same reason as $(b)$; $(e)$ follows from 
\begin{align}
\label{eq:equal_withoutW_11}
0& \leq H(S_{11}^nX_1^n|S^n)-H(S_{11}^nX_1^n|W_2,S^n)=I(S_{11}^nX_1^n;W_2|S^n) \leq I(W_1,S_{11}^nX_1^n;W_2|S^n)\\ \nonumber
& =I(W_1;W_2|S^n)+I(S_{11}^nX_1^n;W_2|W_1,S^n)\stackrel{(k)}{=}I(S_{11}^nX_1^n;W_2|W_1,S^n)\stackrel{(\ell)}{=}0,
\end{align}
where $(k)$ holds since messages and $S^n$ are mutually independent; $(\ell)$ follows the fact that $X_1^n$ is a function of $(W_1,S^n)$. Thus, $H(S_{11}^nX_1^n|S^n)=H(S_{11}^nX_1^n|W_2,S^n)$. Similarly, we can show that
\begin{equation}
\label{eq:equal_withoutW_12}
  H(S_{12}^nX_1^n|S^n)=H(S_{12}^nX_1^n|W_2,S^n);  
\end{equation}
Moreover, $(f)$ holds since, for time instant $1\leq t \leq n$, we have
\begin{align}
H&(S_{12}[t]X_1[t]|S_{12}^{t-1}X_1^{t-1},S^t) =(1-\delta_{12})H(X_1[t]|S_{12}[t]=1,S_{12}^{t-1}X_1^{t-1},S^{t-1}) \\ \nonumber
&+ \delta_{12} \underbrace{H(S_{12}[t]X_1[t]|S_{12}[t]=0,S_{12}^{t-1}X_1^{t-1},S^{t-1})}_{=~0}\\ \nonumber
&\stackrel{(m)}{=}(1-\delta_{12})H(X_1[t]|S_{12}^{t-1}X_1^{t-1},S^t)\\ \nonumber
&\stackrel{(n)}{\ge}(1-\delta_{12})H(X_1[t]|S_{11}^{t-1}X_1^{t-1},S_{12}^{t-1}X_1^{t-1},S^t)\\ \nonumber
&\stackrel{(o)}{=}\frac{1-\delta_{12}}{(1-\delta_{11}\delta_{12})}H(S_{11}[t]X_1[t],S_{12}[t]X_1[t]|S_{11}^{t-1}X_1^{t-1},S_{12}^{t-1}X_1^{t-1},S^t),
\end{align}
where $(m)$ is true because $X_1[t]$ and channel gains are independent at time instant $t$; and $(n)$ is the result of the fact that conditioning reduces entropy; $(o)$ holds since $\Pr[S_{11}[t]=0,S_{12}[t]=0]=\delta_{11}\delta_{12}$. Furthermore, since the transmitted signals at time instant $t$ are independent from channel coefficients in future time instants, we can replace $S^t$ by $S^n$ which leads to
\begin{align}
\sum_{t=1}^n  &H(S_{12}[t]X_1[t]|S_{12}^{t-1}X_1^{t-1},S^n) \\ \nonumber
&\geq \frac{1-\delta_{12}}{(1-\delta_{11}\delta_{12})}\sum_{t=1}^n  H(S_{11}[t]X_1[t],S_{12}[t]X_1[t]|S_{11}^{t-1}X_1^{t-1},S_{12}^{t-1}X_1^{t-1},S^n),
\end{align}
as the result, we obtain
\begin{align}
\label{eq:lemma_r_12,r_11}
H(S_{12}^nX_1^n|S^n) &
\geq \frac{1-\delta_{12}}{(1-\delta_{11}\delta_{12})}H(S_{11}^nX_1^n,S_{12}^nX_1^n|S^n)\\ \nonumber
&\stackrel{(p)}{\geq} \frac{1-\delta_{12}}{(1-\delta_{11}\delta_{12})}H(S_{11}^nX_1^n|S^n)=\frac{1}{\alpha_1}H(S_{11}^nX_1^n|S^n),
\end{align}
where $(p)$ is the result of chain rule;
Furthermore, $(g)$ follows from the fact that conditioning reduces entropy;
$(h)$ holds since $(1-\delta_{22}\delta_{12})$ represents the probability of the case that there is at least one ON link between transmitters and ${\sf Rx}_2$.

Finally, we divide both sides of (\ref{eq:start_out_proof}) by $n$ and let $n \rightarrow \infty$ (i.e. $\xi_n \rightarrow 0$) and obtain
\begin{align}
\label{eq:Final_outer_bound_proof}
\hat{\lambda}_1+\frac{(1-\delta_{11}\delta_{12})}{1-\delta_{12}}\hat{\lambda}_2\leq \frac{(1-\delta_{11}\delta_{12})(1-\delta_{22}\delta_{12})}{1-\delta_{12}}
\end{align}
Similar statement as (\ref{eq:Final_outer_bound_proof}) also holds for $i=2$, which completes the proof.
\end{proof}

\section{Example of control table when $0 \leq \delta \leq \frac{3-\sqrt{5}}{2}$ and $Q_1^{1,2|\emptyset},Q_2^{2|\emptyset}$ are origin queues}
\label{apndx:ex_CT}
In this part, we explain the details of movement rules between origin and destination queues when $0 \leq \delta \leq \frac{3-\sqrt{5}}{2}$ and $Q_1^{1,2|\emptyset},Q_2^{2|\emptyset}$ are origin queues. We provide further details as follows:

\vspace{-4mm}
\begin{table}[ht]
\caption{Determine control table when data bits $a$ and $b$ are available in $Q_{1}^{1,2|\emptyset},Q_{2}^{2|\emptyset}$ and $0 \leq \delta \leq \frac{3-\sqrt{5}}{2}$}
\centering
\begin{tabular}{|c|c|c|c|c|c|c|c|}
\hline
SN & Destination Queues & SN & Destination Queues & SN & Destination Queues & SN & Destination Queues \\[0.5ex]

\hline \hline
1 & $\text{F}_3$ &5 & $Q_{1}^{2|1},Q_{2}^{2|\emptyset}$ & 9 & $Q_{1}^{1,2|\emptyset},Q_{2}^{F}$ & 13 & $Q_{1}^{1,2|\emptyset},Q_{2}^{2|1}$ \\
\hline
2 & $Q_{1}^{2|1},Q_{2}^{F}$ & 6 & $Q_{1}^{F},Q_{2}^{2|\emptyset}$ & 10 & $Q_{1}^{1,2|\emptyset},Q_{2}^{F}$ & 14 & $Q_{1}^{1|2},Q_{2}^{2|\emptyset}$ \\
\hline
3 & $Q_{1}^{1,2|\emptyset},Q_{2}^{F}$ & 7 & $Q_{1}^{1,2|\emptyset},Q_{2}^{2|\emptyset}$ & 11 & $Q_{1}^{1,2|\emptyset},Q_{2}^{F}$ & 15  & $Q_{1}^{1|2},Q_{2}^{2|1}$ \\
\hline
4 & $Q_{1}^{2|1},Q_{2}^{F}$ & 8 & $Q_{1}^{F},Q_{2}^{1,2|\emptyset}$ & 12 & $Q_{1}^{1,2|\emptyset},Q_{2}^{F}$ & 16 & $Q_1^{1,2|\emptyset},Q_2^{2|\emptyset}$ \\[1ex]  
\hline
\end{tabular}
\label{table:transition_N_O}
\vspace{-6mm}
\end{table}

\noindent $\diamond$ SN-1: This cases shows flexible destination queues $\text{F}_3$ in Table~\ref{table:flexible queues}. 

\noindent $\diamond$ SN-2: In this SN, ${\sf Rx}_1$ and ${\sf Rx}_2$ receive $a$ and $a \oplus b$, respectively. Therefore, ${\sf Rx}_1$ can decode $a$ easily and ${\sf Rx}_2$ needs $a$ to decode both $a$ and $b$. Hence, $Q_{1}^{2|1}$ and $Q_{2}^{F}$ are the destination queues.
\noindent $\diamond$ SN-3: In this case, ${\sf Rx}_2$ gets $b$ directly form ${\sf Tx}_2$ and ${\sf Rx}_1$ receives $a\oplus b$ form both transmitters. Thus, both receivers still need to receive $a$ form ${\sf Tx}_1$, and this leads to $a \rightarrow Q_{1}^{1,2|\emptyset}$ and $b \rightarrow Q_{2}^{F}$. 
\noindent $\diamond$ SN-4: Both receivers obtain data bits from direct channels, and $a$ and $b$ are the received bits at ${\sf Rx}_1$ and ${\sf Rx}_2$, respectively. As a result, ${\sf Rx}_1$ does not need to receive any other bits, but ${\sf Rx}_2$ still needs $a$. This means, $a \rightarrow Q_{1}^{2|1}$ and $b \rightarrow Q_{2}^{F}$.

\noindent $\diamond$ SN-5 and SN-6: In these SNs, the received data bits consist of only ${\sf Tx}_1$'bit. Therefore, $b$ stays at its origin queue and then $b\rightarrow Q_{2}^{2|\emptyset}$. In SN-5, ${\sf Rx}_2$ still needs $a$ and $a\rightarrow Q_{1}^{2|1}$ and in SN-6 it is not needed to retransmit $a$ and $a\rightarrow Q_{1}^{F}$.

\noindent $\diamond$ SN-7: At this time, ${\sf Rx}_1$ receives $a \oplus b$ and ${\sf Rx}_2$ does not receive anything. In this case, since there is no side information, the data bits will remain in the origin queues. 

\noindent $\diamond$ SN-8: In this SN, ${\sf Rx}_1$ and ${\sf Rx}_2$ receive $a \oplus b$ and $a$, respectively. We find that by delivering $b$ to both receivers, ${\sf Rx}_2$ gets $b$ and ${\sf Rx}_1$ decodes $a$ using $b$ and $a \oplus b$. Hence, $a\rightarrow Q_{1}^{F}$ and $b\rightarrow Q_{2}^{1,2|\emptyset}$.

\noindent $\diamond$ SN-9, 10, and 11: In these situations, as $a$ is not available as the received bit at ${\sf Rx}_1$, $a$ should be delivered to ${\sf Rx}_1$. Moreover, ${\sf Rx}_2$ receives $b$ in SN-9 and SN-10, and $a \oplus b$ in SN-11. Therefore, it is needed to retransmit $a$ to ${\sf Rx}_2$ and $a\rightarrow Q_{1}^{1,2|\emptyset}$ and $b\rightarrow Q_{2}^{F}$.

\noindent $\diamond$ SN-12: Similar to SN-8, we can show that $Q_1^{1,2|\emptyset},Q_2^{F}$ are destination queues.

\noindent $\diamond$ SN-13: In this SN, only ${\sf Rx}_2$ receives $b$ and then $b\rightarrow Q_{2}^{2|1}$ while $a$ stays at its origin queue.

\noindent $\diamond$ SN-14 and SN-15: In these SNs, ${\sf Rx}_2$ receives $a$ and $a\rightarrow Q_{1}^{1|2}$. Furthermore, both channels from ${\sf Tx}_2$ are off in SN-14 and $b\rightarrow Q_{2}^{2|\emptyset}$. However, ${\sf Rx}_1$ receives $b$ in SN-15 and $b\rightarrow Q_{2}^{2|1}$.

\noindent $\diamond$ SN-16: In this case, all links are off, and $a$ and $b$ stay at origin queues. 

\section{Transmission protocol of homogeneous setting when $\frac{3-\sqrt{5}}{2}< \delta \leq 0.5 $}
\label{apndx:A}
In this appendix, we investigate the general transmission protocol for $\frac{3-\sqrt{5}}{2}< \delta \leq 0.5 $. We find that, in this interval, the destination queues are similar to what we described in Section~\ref{Subsection:CT} except for flexible cases when $P[t]=C$. Authors in \cite{vahid2014capacity} shows that the number of bits in $Q_i^{i|\bar{i}}$ is greater than or equal to the number of bits in $Q_i^{\bar{i}|i}$. Hence, we find the flexible destination queues with $P[t]=C$ to include $Q_i^{\bar{i}|i}$ if it is possible. Suppose $a$ and $b$ are available in the origin queues. Table~\ref{table:flexible queues_p=0.6} shows the flexible destination queues with $P[t]=C$ when $\frac{3-\sqrt{5}}{2}< \delta \leq 0.5 $.

\begin{table}[ht]
\vspace{-2mm}
\caption{Determine  flexible destination queues when $a$ and $b$ are available in origin queues and $\frac{3-\sqrt{5}}{2}< \delta \leq 0.5 $}

\centering
\begin{tabular}{|c|c|c|c|c|c|c|c|}
\hline
Flexible ID & SN & $P[t]=C$ &  Flexible ID & SN & $P[t]=C$ \\[0.5ex]

\hline \hline
$F_1$ & 1 & $Q_{1}^{c_1},Q_{2}^{c_1}$ & $F_5$ & 2 & $Q_{1}^{2|1},Q_{2}^{F}$\\[1ex] 
\hline
$F_2$ & 1 & $Q_{1}^{2|1},Q_{2}^{F}$ & $F_6$ & 3 & $Q_{1}^{F},Q_{2}^{1|2}$ \\
\hline
$F_3$ & 1 & $Q_{1}^{c_1},Q_{2}^{c_1}$ & $F_7$ & 8 & $Q_{1}^{F},Q_{2}^{1|2}$ \\
\hline
$F_4$ & 1 & $Q_1^F,Q_{2}^{1|2}$& $F_8$ & 12 &$Q_{1}^{2|1},Q_{2}^{F}$\\[1ex]  
\hline
\end{tabular}
\label{table:flexible queues_p=0.6}
\end{table}

\noindent $\diamond$ $\text{F}_1$ and $\text{F}_3$: These cases are similar to $\text{F}_1$ and $\text{F}_3$ in Table~\ref{table:flexible queues}, respectively.

\noindent $\diamond$ $\text{F}_2$: According to Table~\ref{table:flexible queues}, there are two choices for destination queues in $\text{F}_2$ as $Q_1^{F},Q_2^{2|1}$ and $Q_1^{2|1},Q_2^{F}$, and we choose $Q_1^{2|1},Q_2^{F}$ since they do not increase the number of bits in $Q_i^{i|\bar{i}}$.

\noindent $\diamond$ $\text{F}_4$: Similar to $\text{F}_2$, by changing the users' labels, we can show $a \rightarrow Q_1^{F}$ and $b \rightarrow Q_2^{1|2}$.

\noindent $\diamond$ $\text{F}_5$: Table~\ref{table:flexible queues} shows that $Q_1^{2|1},Q_2^{F}$ and $Q_1^{F},Q_2^{2|1}$ can be considered as destination queues for this flexible case. Here, we pick $Q_1^{2|1},Q_2^{F}$ because they include $Q_i^{\bar{i}|i}$.

\noindent $\diamond$ $\text{F}_6$: By changing the users' labels in $\text{F}_5$, we find that  $a \rightarrow Q_1^{F}$ and $b \rightarrow Q_2^{1|2}$.

\noindent $\diamond$ $\text{F}_7$: In this case, ${\sf Rx}_1$ and ${\sf Rx}_2$ get $a\oplus b$ and $a$, respectively. We can show that $Q_1^{F},Q_2^{1|2}$ and $Q_1^{1|2},Q_2^{F}$ are two options for destination queues, and $Q_1^{F},Q_2^{1|2}$ do not increase the bits in $Q_i^{i|\bar{i}}$.

\noindent $\diamond$ $\text{F}_8$: 
Similar to $\text{F}_7$, by changing the users' labels, we can show that $a \rightarrow Q_1^{2|1}$ and $b \rightarrow Q_2^{F}$.
\section{Transmission protocol of homogeneous setting when $0.5\leq \delta \leq 1$}
\label{apndx:B}
In this appendix, we focus on the transmission protocol of homogeneous setting when $0.5\leq \delta \leq 1$. In this interval, the cornerpoints are close to each other and become a single cornerpoint when $\frac{\sqrt{5}-1}{2}<\delta \leq 1$. Moreover, our simulation shows that number of bits in $Q_i^{i|\bar{i}}$ is greater than number of bits in $Q_i^{\bar{i}|i}$ when $0.5\leq \delta \leq 1$. As a result, we use transmission protocol in Appendix~\ref{apndx:A} and define destination queues under $P[t]=C$ as the destination queues for all values of $P[t]$ since cornerpoints are close to each other.
\section{Proof of Lemma~\ref{claim:transparent_CT}}
\label{apndx:transparent_CT}
\begin{proof}
We prove this lemma by induction. 

First, we show that the lemma holds at $t=1$. Each node is aware of the number of bits in $Q_i^{i|\emptyset}$, and other queues are empty at the beginning of $t=1$. Thus, the origin queues are known and given the channel knowledge the destination queues are determined;

Second, we show that if the lemma holds for time instant $t$, then it also holds for time instant $t+1$. Suppose, a network node knows the number of available bits in each queue at the end of time instant $t$. At the beginning of time instant $t+1$, the node will know the number of available bits in initial queues from both transmitters (either by tracking the last time instant or through the update). Then,  the node determines the highest priority non-empty queues as the origin queues. At the end of time instant $t+1$, the node uses the origin queues, the CSI, and control table to determine the destination queues. This completes the proof.
\end{proof}
\end{appendices}
\bibliographystyle{IEEEtran}
\bibliography{IEEEabrv,Ref}
\end{document}